\newcommand{\CLASSINPUTtoptextmargin}{0.75in}
\newcommand{\CLASSINPUTbottomtextmargin}{1.01in}
\documentclass[conference]{IEEEtran}
\usepackage{tikz}
\usepackage{url}
\usepackage[ruled]{algorithm2e}
\usepackage{algpseudocode}
\usepackage{makecell,multirow,diagbox,booktabs}
\usepackage{enumerate}
\usepackage{cite}
\usepackage[english]{babel}
\usepackage{graphicx}
\usepackage{amsmath}
\usepackage{amsthm}  
\usepackage{multirow}
\usepackage{threeparttable}
\usepackage{enumitem}
\usepackage{colortbl}
\usepackage{tikz}
\usepackage{steinmetz}
\usepackage{amssymb}
\usepackage{array}
\usepackage{pgfplots}
\usepackage{comment}
\newtheorem{theorem}{Theorem} 
\newtheorem{lemma}[theorem]{Lemma}
\usepackage[skip=12pt]{caption}

\def\sysname{\textsc{LiquiRIS}\xspace}

\def\legacy{\textsc{Legacy}\xspace}

\def\methodA{\textsc{LiquiRIS-Single}\xspace}
\def\methodB{\textsc{LiquiRIS-Multi}\xspace}

\usepackage{xcolor}

\definecolor{DarkGreen}{RGB}{0,150,0}

\definecolor{Orange}{RGB}{245,100,10}

\usepackage[acronym]{glossaries}
 
\newacronym{AoA}{AoA}{angle of arrival}
\newacronym{AoD}{AoD}{angle of departure}
\newacronym{DAC}{DAC}{digital to analog converter}
\newacronym{DC}{DC}{direct current}
\newacronym{DRIE}{DRIE}{deep reactive Ion etching}
\newacronym{ISML}{IMSL}{inverted microStrip line}
\newacronym{LCD}{LCD}{liquid crystal display}
\newacronym{LoS}{LoS}{line-of-sight}
\newacronym{LC}{LC}{liquid crystal}
\newacronym{LC-RIS}{LC-RIS}{}
\newacronym{NLoS}{NLoS}{non-line-of-sight}
\newacronym{MEMS}{MEMS}{micro-electro-mechanical systems}
\newacronym{NLC}{NLC}{nematic liquid crystal}
\newacronym{PIN}{PIN}{positive-intrinsic-negative}
\newacronym{RF}{RF}{radio frequency}
\newacronym{RIS}{RIS}{reconfigurable intelligent surface}
\newacronym{SNR}{SNR}{signal-to-noise ratio}
\newacronym{SINR}{SINR}{signal-to-noise-plus-interference ratio}
\newacronym{TFT}{TFT}{thin-film transistor}
\newacronym{NRCS}{NRCS}{normalized RIS cross-section}
\newacronym{AWGN}{AWGN}{additive white Gaussian noise}

\newacronym{CSI}{CSI}{channel state information}
\newacronym{MILP}{MILP}{mixed-integer linear program}
\newacronym{CDF}{CDF}{cumulative distribution function}
\newacronym{MT}{MT}{mobile terminal}
\newacronym{BS}{BS}{base station}

\usepackage{soul} 
\usepackage{subcaption}
\usepackage{titlesec}
\titleformat{\subsubsection}[runin]
{\normalfont\bfseries}{\thesubsubsection}{1em}{}[:$~$]
\usepackage{siunitx}
\DeclareSIUnit{\belisotropic}{Bi}
\DeclareSIUnit{\dBi}{\deci\belisotropic}
\DeclareSIUnit{\bit}{bit}

\allowdisplaybreaks

\makeatletter
\let\save@mathaccent\mathaccent
\newcommand*\if@single[3]{%
  \setbox0\hbox{${\mathaccent"0362{#1}}^H$}%
  \setbox2\hbox{${\mathaccent"0362{\kern0pt#1}}^H$}%
  \ifdim\ht0=\ht2 #3\else #2\fi
  }
\newcommand*\rel@kern[1]{\kern#1\dimexpr\macc@kerna}
\newcommand*\widebar[1]{\@ifnextchar^{{\wide@bar{#1}{0}}}{\wide@bar{#1}{1}}}
\newcommand*\wide@bar[2]{\if@single{#1}{\wide@bar@{#1}{#2}{1}}{\wide@bar@{#1}{#2}{2}}}
\newcommand*\wide@bar@[3]{%
  \begingroup
  \def\mathaccent##1##2{%
    \let\mathaccent\save@mathaccent
    \if#32 \let\macc@nucleus\first@char \fi
    \setbox\z@\hbox{$\macc@style{\macc@nucleus}_{}$}%
    \setbox\tw@\hbox{$\macc@style{\macc@nucleus}{}_{}$}%
    \dimen@\wd\tw@
    \advance\dimen@-\wd\z@
    \divide\dimen@ 3
    \@tempdima\wd\tw@
    \advance\@tempdima-\scriptspace
    \divide\@tempdima 10
    \advance\dimen@-\@tempdima
    \ifdim\dimen@>\z@ \dimen@0pt\fi
    \rel@kern{0.6}\kern-\dimen@
    \if#31
      \overline{\rel@kern{-0.6}\kern\dimen@\macc@nucleus\rel@kern{0.4}\kern\dimen@}%
      \advance\dimen@0.4\dimexpr\macc@kerna
      \let\final@kern#2%
      \ifdim\dimen@<\z@ \let\final@kern1\fi
      \if\final@kern1 \kern-\dimen@\fi
    \else
      \overline{\rel@kern{-0.6}\kern\dimen@#1}%
    \fi
  }%
  \macc@depth\@ne
  \let\math@bgroup\@empty \let\math@egroup\macc@set@skewchar
  \mathsurround\z@ \frozen@everymath{\mathgroup\macc@group\relax}%
  \macc@set@skewchar\relax
  \let\mathaccentV\macc@nested@a
  \if#31
    \macc@nested@a\relax111{#1}%
  \else
    \def\gobble@till@marker##1\endmarker{}%
    \futurelet\first@char\gobble@till@marker#1\endmarker
    \ifcat\noexpand\first@char A\else
      \def\first@char{}%
    \fi
    \macc@nested@a\relax111{\first@char}%
  \fi
  \endgroup
}
\makeatother

\begin{document}

\title{Fast-Reconfiguring Liquid-Crystal RIS for Pervasive Wireless Networks\\[-4mm]}

\author{
	\IEEEauthorblockN{
	Luis F. Abanto-Leon\IEEEauthorrefmark{1}\textsuperscript{\textsection}, 
	Robin Neuder\IEEEauthorrefmark{2}\textsuperscript{\textsection}, 
	Waqar Ahmed\IEEEauthorrefmark{5}, 
	Alejandro Jimenez Saez\IEEEauthorrefmark{2}, 
	Vahid Jamali\IEEEauthorrefmark{2}, 
	Arash Asadi\IEEEauthorrefmark{3}
	}
	\IEEEauthorblockA{
	\IEEEauthorrefmark{1}Ruhr University Bochum, 
	\IEEEauthorrefmark{2}TU Darmstadt, 
	\IEEEauthorrefmark{5}Microchip Technology, 
	\IEEEauthorrefmark{3}TU Delft
	} 
}

\maketitle
\begingroup\renewcommand\thefootnote{\textsection}
\footnotetext{Both authors contributed equally to this research.}
\endgroup

\maketitle


\begin{abstract}

Reconfigurable intelligent surfaces (RISs) have emerged as a key technology for dynamically reshaping wireless propagation, enhancing coverage and mitigating blockages to enable more pervasive network connectivity. However, implementing RISs at high frequencies remains challenging due to the cost and power demands of semiconductor-based components. To address these critical limitations, liquid crystals (LCs) technology has been identified as a promising low-cost and low-power alternative, giving rise to LC-RIS. The central challenge of this technology, however, lies in its limited responsiveness, as the slow molecular dynamics of LCs lead to long phase-shift reconfiguration times that restrict practicality.

This paper presents \sysname, a novel framework that enables substantially faster phase shifting in LC-RIS. By explicitly incorporating the physical dynamics of LC molecules into the phase-shift configuration process, \sysname intelligently selects phase transitions that minimize the overall reconfiguration time. As a result, \sysname achieves up to $ 71.61 \% $ reduction in overall reconfiguration time compared to conventional schemes, significantly improving the feasibility of LC-RIS deployment. The proposed framework is further validated through experiments on a mmWave LC-RIS prototype.

\end{abstract}



\widowpenalty=10
\clubpenalty=10
\brokenpenalty=10


\glsresetall


\section{Introduction}

Reconfigurable intelligent surfaces (RISs) have emerged as a promising technology for dynamically controlling radio frequency (RF) wave propagation through adjustment of reflected signal properties. RISs are envisioned to enhance wireless connectivity by improving signal strength, extending coverage, increasing capacity, and mitigating blockages. This capability to enable seamless and pervasive wireless environments, from homes and IoT networks to public spaces and smart cities, has attracted significant research interest.

The majority of existing RIS designs implement phase shifting using semiconductor components such as PIN diodes~\cite{rossanese2022designing, gros2021reconfigurable, zeng2021high, amri2021reconfigurable}, varactors~\cite{pei2021ris, araghi2022reconfigurable, sievenpiper2003two}, and MEMS devices~\cite{liu2022terahertz}, with most demonstrations at sub-6 GHz~\cite{yezhen2020novel, trichopoulos2021design, dai2020reconfigurable, VincentPoor2020, RomainDiRenzo2021,arun2020rfocus, scattermimo, rossanese2022designing} and a few at mmWave frequencies~\cite{gros2021reconfigurable, tan2018enabling}, showing promising results. However, large-scale deployment of semiconductor-based RISs, particularly at higher frequencies, remains hindered by intrinsic drawbacks such as high power consumption, high fabrication costs, limited scalability, and coarse phase resolution. \emph{These limitations have motivated the exploration of alternative technologies that are inherently more energy-efficient, cost-effective, and scalable.}

\textbf{Opportunities.} Liquid crystals (LCs) have emerged as a promising alternative to semiconductor components, offering a viable path to overcoming the limitations of conventional RIS designs~\cite{neuder2023compact, wu2020liquid}. LC-RISs consume minimal power and can be fabricated at substantially lower cost, which remains nearly invariant across frequencies. LC-RISs are amenable to large-scale manufacturing and inherently support fine-grained phase control~\cite{jimenez2023reconfigurable}. \emph{Collectively, these characteristics position LC-RISs as a promising technology for sustainable and cost-effective RIS deployments.}

\begin{figure}[!t]
	\begin{center}
	 	\begin{subfigure}{0.43\columnwidth}
				\includegraphics[width=1\columnwidth]{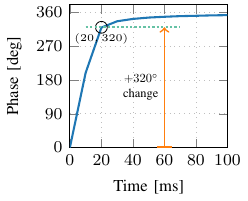} 
				\caption{\centering Positive phase change}
				\label{figure_phase_transition_time_a}
	 	\end{subfigure}
	  	\hspace{5mm}
	 	\begin{subfigure}{0.43\columnwidth}
				\includegraphics[width=1\columnwidth]{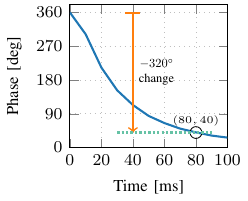} 
				\caption{\centering Negative phase change}
				\label{figure_phase_transition_time_b}
		\end{subfigure}
		\vspace{-2mm}
	    \caption{Response time of an LC-RIS unit-cell, illustrating that positive phase changes occur faster than negative ones. \emph{For example, a positive phase change of $320^\circ$ completes in $\approx 20$ ms, whereas a negative change of the same magnitude takes $\approx 80$ ms.}}
	 	\label{figure_phase_transition_time}
 	\end{center}
 	\vspace{-7mm}
\end{figure}

\textbf{Challenges.} Despite these advantages, the adoption of LC-RISs is constrained by their limited responsiveness. Phase shifting in LC-RISs is realized through physical reorientation of LC molecules, a process that is substantially slower than electronic switching in semiconductor-based counterparts. Moreover, this reorientation process is highly asymmetric, with negative phase-shift changes taking considerably longer than positive ones, as illustrated in Fig. \ref{figure_phase_transition_time}, which further exacerbates the overall \emph{reconfiguration time} of LC-RISs and hinders their ability to rapidly adapt.\footnote{In this work, \emph{reconfiguration time} refers to the total duration required by an LC-RIS to transition from one phase-shift configuration to another, i.e., the time taken for all LC unit-cells to reorient their molecules to the target states. This is distinct from \emph{response time}, an established term which describes the same transition duration, but for a single LC unit-cell.}

Existing studies have attempted to mitigate this issue by developing new LC composites \cite{guirado2022mm}, reducing LC layer thickness \cite{neuder2024architecture}, or using overshooting techniques~\cite{yang_fundamentals_2006} that apply high voltages to accelerate molecular reorientation. Although these hardware-oriented solutions offer partial improvements, they remain constrained by stability concerns, fabrication precision, and increased cost. \emph{These limitations highlight the need for complementary approaches that reduce LC-RIS reconfiguration time without relying solely on hardware advancements.}

\textbf{Solution.} We propose \sysname, a novel algorithmic approach that incorporates the intrinsic dynamic behavior of LC materials into the phase-shift optimization process. Conventional schemes, designed for semiconductor-based RISs, are agnostic to these dynamics, leading to prohibitively long reconfiguration times. \emph{By explicitly accounting for LC molecular reorientation dynamics, \sysname identifies phase-shift configurations that meet communication requirements with minimal reconfiguration time, thereby enabling rapid and efficient LC-RIS adaptation.}


{\bf Contributions.} Our main contributions are:
\begin{itemize}
\itemsep0em 

	\item Leveraging empirical measurements from an LC unit-cell, we develop a piecewise analytical model that accurately captures the asymmetric dynamics governing phase change and response time. This model is highly accurate, with a total error of less than $ 0.03\% $ compared to ground-truth measurements.

	\item Building on this model, we formulate the problem of minimizing LC-RIS reconfiguration time subject to communication constraints. We propose two convex optimization-based algorithms, \methodA and \methodB, both expressed as tractable mixed-integer linear program (MILP). The first algorithm minimizes reconfiguration time for sequential user serving, while the second optimizes simultaneously a sequence of multiple configurations, thereby achieving reduced overall reconfiguration time. The proposed algorithms reduce the reconfiguration time by up to $ 71.61\% $ and $ 64.36 \% $, respectively, compared to conventional schemes.
                       
	\item We investigate the impact of channel state information (CSI) accuracy on reconfiguration time and demonstrate that imperfect CSI significantly prolongs it, revealing a previously unreported effect.
	
	\item We experimentally validate our approach on an LC-RIS prototype operating at mmWave frequencies. 
            
\end{itemize}

\textbf{Related work.} While numerous studies have sought to mitigate the long response times of LC-based designs through hardware innovations~\cite{guirado2022mm, neuder2024architecture, chien2017fast, guo2020line, song2012ultrafast}, only a few have explored algorithmic approaches. For instance, \cite{delbari2024fast, delbari2025fast} proposed penalizing negative phase transitions more heavily than positive ones to account for the asymmetric dynamics illustrated in Fig \ref{figure_phase_transition_time}. Although intuitive, these approaches relied on heuristic weighting factors and therefore fail to accurately capture the underlying LC dynamics. In contrast, our work develops a data-driven piecewise analytical model that closely fits empirical LC behavior, ensuring high fidelity. Furthermore, we investigate the co-optimization of multiple phase-shift configurations, an aspect unexplored in prior art. In addition, we uncover the interplay between CSI inaccuracy and reconfiguration time, revealing a critical limiting factor.

\emph{Notation}: Matrices and vectors are denoted by $ \mathbf{X} $ and $ \mathbf{x} $, respectively. The transpose and Hermitian transpose of $ \mathbf{X} $ are denoted by $ \mathbf{X}^\mathrm{T} $ and $ \mathbf{X}^\mathrm{H} $, respectively. Also, $ \mathrm{j} \triangleq \sqrt{-1} $ is the imaginary unit, $ \mathbb{E} \left\lbrace \cdot \right\rbrace  $ denotes statistical expectation, and $ \mathcal{CN} \left( \upsilon, \xi^2 \right) $ denotes the complex Gaussian distribution with mean $ \upsilon $ and variance $ \xi^2 $. The set of complex, real, and nonnegative real numbers are denoted by $ \mathbb{C} $, $ \mathbb{R} $, and $ \mathbb{R}_{+} $, respectively.



\section{System Model and Problem Formulation} 
\label{section_approach}
\label{section_SysModel}


\subsection{System Model}

We consider a system comprising a base station (BS), a set of $ L $ mobile terminals (MTs), and a planar LC-RIS that is located in the far-field relative to the BS and MTs. The LC-RIS is centered at location $ (0,0,0) $ in the $xz$-plane and has $ N_\mathrm{x} $ unit-cells per row and $ N_\mathrm{z} $ unit-cells per column, making a total of $ N = N_\mathrm{x} N_\mathrm{z} $ unit-cells. Each of the unit-cells has a constant modulus and a tunable phase. 
\begin{figure}[!t]
	\centering
	\includegraphics[width=0.97\columnwidth]{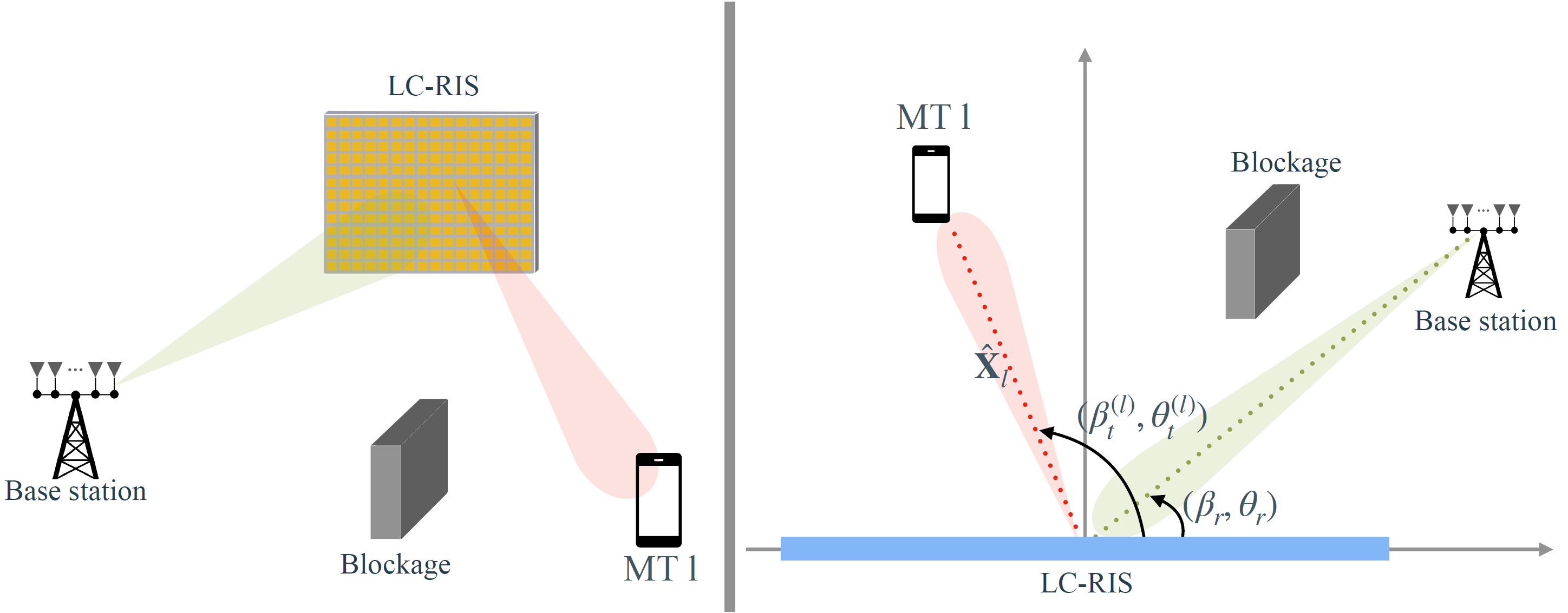} 
	\caption{Overview of our system model.}
	\label{figure_system_model}
\end{figure}

The BS is equipped with a set of predefined beams, each oriented toward a different angle relative to the LC-RIS. In the event of a blockage between the BS and MTs, as illustrated in~Fig. \ref{figure_system_model}, the BS selects one of these beams to serve the MTs, using the LC-RIS as a reflector. We assume that the selected beam remains fixed over a transmission interval during which the MTs are served. Since the BS serves one MT at a time, the LC-RIS must be reconfigured to reflect the received signals toward the corresponding $l $-th MT, where $ l \in \mathcal{L} = \left\lbrace 1, \dots, L \right\rbrace $.

The channel between the BS and the LC-RIS follows a Rician model, expressed as
\begin{align}
	\mathbf{g} = \rho \left( \sqrt{\frac{\Lambda}{\Lambda+1}} \mathbf{g}_\mathrm{LoS} + \sqrt{\frac{1}{\Lambda+1}} \mathbf{g}_\mathrm{NLoS} \right) ,
\end{align}
where $ \rho $ denotes the large-scale fading coefficient, $ \Lambda $ is the Rician factor, $ \mathbf{g}_\mathrm{LoS} = \big[ e^{-\mathrm{j} \frac{2 \pi}{\lambda} \psi_1}, \dots, e^{-\mathrm{j} \frac{2 \pi}{\lambda} \psi_N} \big]^\mathrm{T} $ represents the line-of-sight (LoS) component, and $ \mathbf{g}_\mathrm{NLoS} \sim \mathcal{CN} \left( \mathbf{0}, \mathbf{I} \right) $ models the non-line-of-sight (NLoS) scattering components. Here, $ \psi_n = \mathbf{v}_\mathrm{r}^\mathrm{T} \mathbf{p}_n $, where $ \mathbf{v}_\mathrm{r} = \left[ \sin(\theta_\mathrm{r}) \cos(\beta_\mathrm{r}), \sin(\theta_\mathrm{r}) \sin(\beta_\mathrm{r}), \cos(\theta_\mathrm{r}) \right]^\mathrm{T} \in \mathbb{R}^{3 \times 1} $ captures the phase variations of the incident wave over the LC-RIS, and $\mathbf{p}_n \in \mathbb{R}^{3 \times 1}$ denotes the Cartesian coordinates of the $n$-th unit-cell, with $ n \in \mathcal{N} = \left\lbrace 1, \dots, N \right\rbrace $. The parameters $\beta_\mathrm{r}$ and $\theta_\mathrm{r}$ correspond to the azimuth and elevation angles of arrival (AoA) of the signal impinging on the LC-RIS, respectively. 

Similarly, the channel between the LC-RIS and the $l$-th MT is expressed as
\begin{align}
	\mathbf{h}_l = \mu^{(l)} \left( \sqrt{\frac{\Lambda_l}{\Lambda_l+1}} \mathbf{h}_{\mathrm{LoS},l} + \sqrt{\frac{1}{\Lambda_l+1}} \mathbf{h}_{\mathrm{NLoS},l}  \right) ,
\end{align}
where $ \mu^{(l)} $ denotes the large-scale fading coefficient, $ \Lambda_l $ is the Rician factor associated with the $ l $-th MT, $ \mathbf{h}_{\mathrm{LoS},l} = \big[ e^{-\mathrm{j} \frac{2 \pi}{\lambda} \xi_1^{(l)}}, \dots, e^{-\mathrm{j} \frac{2 \pi}{\lambda} \xi_N^{(l)}} \big]^\mathrm{T} $ represents the LoS component, and $ \mathbf{h}_{\mathrm{NLoS},l} $ captures the NLoS scattering contribution. Here, $ \xi_n^{(l)} = {\mathbf{v}_\mathrm{t}^{(l)}}^\mathrm{T} \mathbf{p}_n $, with $ \mathbf{v}_\mathrm{t}^{(l)} = \big[ \sin(\theta_\mathrm{t}^{(l)}) \cos(\beta_\mathrm{t}^{(l)}), \sin(\theta_\mathrm{t}^{(l)}) \sin(\beta_\mathrm{t}^{(l)}),  \cos(\theta_\mathrm{t}^{(l)}) \big]^\mathrm{T} \in \mathbb{R}^{3 \times 1} $ characterizing the phase variations of the reflected wave. The parameters $ \beta_\mathrm{t}^{(l)} $ and $ \theta_\mathrm{t}^{(l)} $ denote the azimuth and elevation angles of departure (AoDs) for the $l$-th MT, respectively.

The phase-shift configuration for the $l$-th MT is defined as
\begin{align}
	\hat{\mathbf{x}}_l = \delta e^{\mathrm{j} \boldsymbol{\phi}_l} \in \mathbb{C}^{N \times 1},
\end{align}
where $\boldsymbol{\phi}_l = [\phi_{l,1}, \dots, \phi_{l,N}]^\mathrm{T} \in [0, 2\pi)^{N \times 1}$ represents the set of phase shifts applied across all LC-RIS unit-cells, $ \phi_{l,n} $ is the $n$-th phase shift, and $\delta$ denotes the unit-cell reflection coefficient.

Let $\hat{s}_l \in \mathbb{C} $ denote the signal transmitted to the $l$-th MT, such that  $ \mathbb{E} \left\lbrace \hat{s}_l \right\rbrace = 0 $ and $ \mathbb{E} \left\lbrace |\hat{s}_l|^2 \right\rbrace = 1 $. Thus, the received signal at the $l$-th MT is then given by
\begin{align} 
	\begin{aligned}
	y_l & = \sqrt{P_\mathrm{BS} G_\mathrm{BS} G_\mathrm{MT}} \mathbf{h}_l^\mathrm{H} \mathrm{diag} (\hat{\mathbf{x}}_l ) \mathbf{g}^{*} \hat{s}_l + \eta_l,
	\\ 
	    & = \sqrt{P_\mathrm{BS} G_\mathrm{BS} G_\mathrm{MT}} \mathbf{h}_l^\mathrm{H} \mathrm{diag} (\delta e^{\mathrm{j} \boldsymbol{\phi}_l} ) \mathbf{g}^{*} \hat{s}_l + \eta_l,  
	    \\ 
	    & = \delta \sqrt{P_\mathrm{BS} G_\mathrm{BS} G_\mathrm{MT}} \mathbf{h}_l^\mathrm{H} \mathrm{diag} ( \mathbf{g}^{*} ) e^{\mathrm{j} \boldsymbol{\phi}_l} \hat{s}_l + \eta_l,
	\end{aligned}
\end{align}
where $P_\mathrm{BS}$ denotes the BS transmit power, $ G_\mathrm{BS} $ and $ G_\mathrm{MT} $ represent the antenna gains at the BS and MT, respectively, and $ \eta_l \sim \mathcal{CN} \left( 0, \sigma^2 \right) $ accounts for additive white Gaussian noise (AWGN). Therefore, the signal-to-noise ratio (SNR) at the $l$-th MT is defined as
\begin{align*} 
	\mathsf{SNR} \left( \boldsymbol{\phi}_l \right) =  \frac{\delta^2 K^2}{\sigma^2} \left| \mathbf{h}_l^\mathrm{H} \mathrm{diag} ( \mathbf{g}^{*} ) e^{\mathrm{j} \boldsymbol{\phi}_l} \right|^2, \forall l \in \mathcal{L},
\end{align*}
where $ K^2 = P_\mathrm{BS} G_\mathrm{BS} G_\mathrm{MT} $.

\subsection{Problem formulation} \label{sec:optimization-problems}

We formulate two optimization problems aimed at minimizing the reconfiguration time. The first problem, termed \methodA, is a myopic optimizer that determines the next phase-shift configuration. The second, \methodB, is a foresighted optimizer that jointly determines an entire sequence of phase-shift configurations, thereby capturing interdependencies among successive configurations and minimizing even further the overall cumulative reconfiguration time compared to \methodA.

\textbf{\methodA:}
To serve the $l$-th MT, the LC-RIS must transition from configuration $\boldsymbol{\phi}_{l-1}$ to $\boldsymbol{\phi}_l$. The overall LC-RIS reconfiguration time is determined by the slowest unit-cell, i.e., the one requiring the longest time to complete its phase update. In addition, each MT must satisfy a minimum SNR requirement. Considering these factors, we formulate 
\begin{align*} 
	\mathcal{P}^{(l)}_1: & ~~ \underset{ \boldsymbol{\phi}_l}{\mathrm{minimize}}
	& & 
	\max_{n \in \mathcal{N} } \hat{f} \left( \phi_{l,n} - \phi_{l-1,n} \right) 
	\\
	& ~~~~~~ \mathrm{s.t.} & \mathrm{C}_{1}: ~ & \mathsf{SNR} \left( \boldsymbol{\phi}_l \right) \geq \Gamma_{\mathrm{th},l}, 
	\\
	& & \mathrm{C}_2: ~ & \phi_{l,n} \in [0, 2\pi), \forall n \in \mathcal{N},  
\end{align*}
where function $ \hat{f}(\cdot) $ models the response time of an individual unit-cell as a function of its phase change.\footnote{The response-time function $ \hat{f} (\cdot) $ is characterized in Section \ref{sec:piecewise-function}.} Consequently, the objective seeks to minimize the overall reconfiguration time, determined by the slowest-responding unit cell. Constraint $ \mathrm{C}_{1} $ enforces a minimum SNR requirement of $ \Gamma_{\mathrm{th},l} $ for the $l$-th MT and constraint $ \mathrm{C}_{2} $ bounds the phase values to $ [0, 2\pi) $.

\textbf{\methodB:} This problem extends \methodA\ by jointly optimizing the phase-shift configurations for all $ L $ MTs, accounting for the cumulative response time required to transition sequentially between configurations (i.e., from $ \boldsymbol{\phi}_1 $ to $ \boldsymbol{\phi}_2 $, then $ \boldsymbol{\phi}_2 $ to $ \boldsymbol{\phi}_3 $, and so on), as shown below
\begin{align*} 
	\mathcal{P}_2: & ~~ \underset{ \boldsymbol{\phi}_1, \dots, \boldsymbol{\phi}_L}{\mathrm{minimize}}
	& &  
	\sum_{l \in \mathcal{L}} \max_{n \in \mathcal{N} } \hat{f} \left( \phi_{l,n} - \phi_{l-1,n} \right) 
	\\
	& ~~~~~~ \mathrm{s.t.} & \widebar{\mathrm{C}}_{1}: ~ & \mathsf{SNR} \left( \boldsymbol{\phi}_l \right) \geq \Gamma_{\mathrm{th},l}, \forall l \in \mathcal{L},
	\\
	& & \widebar{\mathrm{C}}_{2}: ~ & \phi_{l,n} \in [0, 2\pi), \forall l \in \mathcal{L}, n \in \mathcal{N}, 
\end{align*}
where the objective minimizes the total reconfiguration time, defined as the sum of the slowest unit-cell response times across all consecutive configurations in the sequence.

\section{Response Time Model} \label{sec:response-time-model}


\subsection{Primer on LC technology} \label{sec:primer-lc-technology}

The fundamental building block of an LC-RIS is a tunable unit-cell, as illustrated in Fig. \ref{fig:LCCharacteristics1}. At its core lies a layer of LC material, whose effect on an RF wave is governed by the orientation of the LC molecules relative to the wave's electric field. This molecular orientation determines the permittivity of the LC material, thereby controlling the phase shift imparted to the RF wave, as described in the following.
\begin{figure}[!h]
	\begin{center}
	 	\begin{subfigure}{0.48\columnwidth}
				\includegraphics[width=0.9\textwidth]{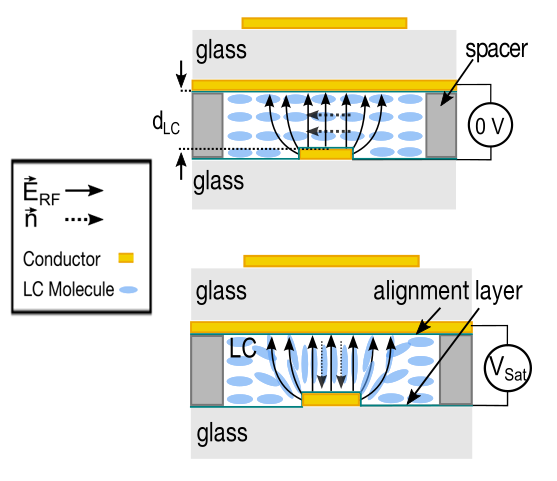}
				\caption{LC unit-cell}
				\label{fig:LCCharacteristics1}
	 	\end{subfigure}
	 	\begin{subfigure}{0.48\columnwidth}
				\includegraphics[width=0.94\textwidth]{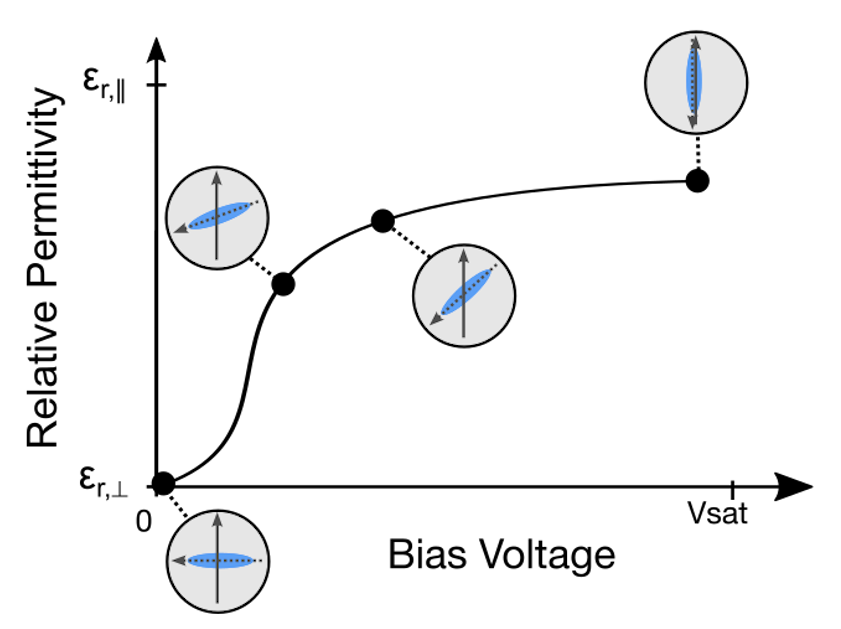}
				\caption{LC permittivity}
				\label{fig:LCCharacteristics2}
		\end{subfigure}
	    \caption{LC principles.}
	    \vspace{-1mm}
	 	\label{fig:LCCharacteristics}
 	\end{center}
 	\vspace{-4mm}
\end{figure}

$ \bullet $ \textit{Low permittivity:} In this regime, the RF wave propagates faster and experiences a smaller phase shift. This occurs when the long molecular axis $ \vec{n} $ is perpendicular to the wave's electric field $ \vec{E}_\mathrm{RF} $ (Fig. \ref{fig:LCCharacteristics1}, top), as the molecules in this orientation interact minimally with the incoming field.

$ \bullet $ \textit{High permittivity:} In this regime, the RF wave propagates more slowly, resulting in a larger phase shift. This occurs when the long molecular $ \vec{n} $ is aligned parallel to $ \vec{E}_\mathrm{RF} $ (Fig. \ref{fig:LCCharacteristics1}, bottom), as the molecules in this orientation interact more strongly with the wave's electric field $ \vec{E}_\mathrm{RF} $.

\textbf{Low to high permittivity variation:} The application of a bias voltage induces a relatively quick transition from a low- to a high-permittivity state, which corresponds to the \emph{positive phase-change} regime shown in Fig.~\ref{figure_phase_transition_time_a}. When no bias is applied, the LC molecules are oriented perpendicular to $ \vec{E}_\mathrm{RF} $, resulting in a lower relative permittivity $ \varepsilon_\mathrm{r,\perp} $, as illustrated in Fig.~\ref{figure_phase_transition_time_b}. As the bias voltage increases, the molecules gradually tilt toward $ \vec{E}_\mathrm{RF} $. At the saturation voltage $ V_\mathrm{sat} $, they become fully aligned with the field, yielding a higher relative permittivity $ \varepsilon_\mathrm{r,\parallel} $. Thus, by varying the bias voltage between $ 0 $ and $ V_\mathrm{sat} $, the permittivity, and consequently the phase of the traversing RF wave, can be precisely controlled.

\textbf{High to low permittivity variation:} Conversely, the transition from high to low permittivity is significantly slower, corresponding to the \emph{negative phase-change} regime shown in Fig.~\ref{figure_phase_transition_time_b}. This asymmetry arises because, once the bias is removed, the LC molecules rely solely on elastic restoring forces to return to their resting positions. The resulting slow relaxation process constitutes the primary bottleneck to achieving rapid LC-RIS reconfiguration.


\subsection{Piecewise function for response time characterization} \label{sec:piecewise-function}

We characterize the LC unit-cell's response time through an experimental campaign, collecting $400$ measurements that reveal the ground-truth relationship between phase change and response time.

However, directly exploiting this dataset to compute the response time corresponding to a given phase-shift change is impractical, as the measurements exhibit small oscillations and irregularities. To address this issue, we develop a piecewise linear approximation that constructs a convex function $ \hat{f}: \mathcal{R} \rightarrow \mathbb{R}_{+} $ closely matching the empirical data. This representation is particularly suitable as it yields a tractable analytical form compatible with our convex optimization framework employed to solve problems $ \mathcal{P}^{(l)}_1 $ and $ \mathcal{P}_2 $, formulated in Section \ref{sec:optimization-problems}.

We select $I$ representative samples from the collected dataset and define a piecewise linear function $\hat{f}$, constructed to be convex, defined as follows
\begin{equation} \label{equation_piecewise_linear_function}
	\hat{f} \left( \phi \right) = 
	\begin{cases}
		   	\hat{f}_{1} \left( \phi \right) \triangleq a_1 \phi + b_1, & \phi \in \Phi_1,
		   	\\
		   	~~~~~~~~~~~~~~~~ \vdots
		   	\\	
		   	\hat{f}_{I-1} \left( \phi \right) \triangleq a_{I-1} \phi + b_{I-1}, & \phi \in \Phi_{I-1},
	\end{cases}
\end{equation}
where $ \phi $ denotes the phase change, which can be positive (counter-clockwise) or negative (clockwise). Additionally, $ a_i \in \mathbb{R} $, and $ b_i \in \mathbb{R} $, $ \forall i \in \mathcal{I} $, denote respectively the slope and vertical shift of $ \hat{f}_{i} $ with respect to the origin, where $ \mathcal{I} = \left\lbrace 1, \dots, I - 1 \right\rbrace $. The slope and vertical shift of $ \hat{f}_{i} $ are defined as $ a_i = \frac{w_i-w_{i+1}}{c_i-c_{i+1}} $ and $ b_i = w_{i+1} - \frac{w_i-w_{i+1}}{c_i-c_{i+1}} c_{i+1} $, such that $ c_i \in \mathbb{R} $ and $ w_i \in \mathbb{R}_{+} $ denote respectively the breakpoints in the abscissa and ordinate at which $ \hat{f}_{i} $ and $ \hat{f}_{i+1} $ intersect. Besides, $ \Phi_i = \left[ c_{i}, c_{i+1} \right] $ represents the interval in which function $ \hat{f}_{i} $ prevails over other functions $ \hat{f}_{j \neq i} $.

As established in \textbf{Lemma \ref{lemma_2}}, function $\hat{f}$ can equivalently be expressed as
\begin{equation} \label{eq:piecewise-approximation}
	\hat{f} \left( \phi \right) = \max_{i \in \mathcal{I}} \hat{f}_i \left( \phi \right).
\end{equation}

\begin{lemma} \label{lemma_2}
	Formulations (\ref{equation_piecewise_linear_function}) and (\ref{eq:piecewise-approximation}) are equivalent for the convex piecewise linear function $ \hat{f} (\phi) $.
\end{lemma}
\begin{proof}
	Please, refer to Appendix \ref{app:proof-lemma-1}.
\end{proof}

The proposed modeling framework is general and can be configured with an arbitrary number of segments. As an illustration, Fig.~\ref{figure_approximated_phase_transition_time} compares the ground truth with a piecewise-linear approximation using $ I = 17 $ ($16$ segments). Despite the relatively small number of segments, the approximation error remains below $0.03\%$ with respect to the ground truth, making it an accurate and reliable representation\footnote{Note that the data labeled as ``ground truth'' in Fig.~\ref{figure_approximated_phase_transition_time} corresponds to a transposed version of the data in Fig.~\ref{figure_phase_transition_time}, i.e., with the axes interchanged.}. Consequently, without loss of generality, this approximation is adopted throughout this work, with the corresponding parameters $ w_i $ and $ c_i $ provided in Appendix~\ref{app:parameters-piecewise-function}.

\begin{figure}[!h]
	\centering
	\includegraphics[width=0.8\columnwidth]{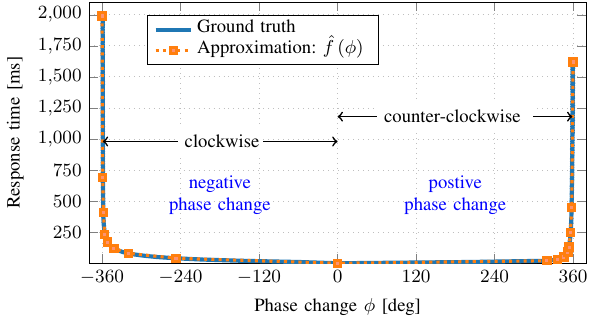} 
	\vspace{-2mm}
	\caption{Ground truth and approximate response time.}
	\label{figure_approximated_phase_transition_time}
	\vspace{-2mm}
\end{figure}

\section{Optimization Framework}

In this section, we first present the solution to \methodA and then extend it to solve \methodB, leveraging the structural similarity between the two problems.

\subsection{\methodA} 
\label{section_methodA}

Problem $ \mathcal{P}^{(l)}_1 $ is a nonconvex nonlinear program and is challenging to solve. We devise an algorithm based on convex optimization principles to solve $ \mathcal{P}^{(l)}_1 $. Specifically, we introduce a sequence of convexification steps designed to systematically handle the nonconvexities of $ \mathcal{P}^{(l)}_1 $.

\subsubsection{Introducing auxiliary decision variables}

We introduce auxiliary decision variables $ \gamma_{l,n} \in \mathbb{R} $ and $ \tau_l \in \mathbb{R}_{+} $ to facilitate the reformulation of $ \mathcal{P}^{(l)}_1 $. Specifically, $ \gamma_{l,n} $ represents the argument of the function $ \hat{f} $, leading to the constraint $ \mathrm{C}_{3}: \gamma_{l,n} = \phi_{l,n} - \phi_{l-1,n} $, $\forall n \in \mathcal{N} $. Moreover, we employ the nonnegative auxiliary variable $ \tau_l $ to bound the maximum response time associated with the $l$-th phase-shift configuration. Accordingly, we introduce constraint $ \mathrm{C}_{4}: \tau_l \geq 0 $, and enforce $ \mathrm{C}_{5}: \max_{n \in \mathcal{N}} \hat{f}(\gamma_{l,n}) \leq \tau_l $, such that the objective function becomes $ \tau_l $. Incorporating these modifications, problem $ \mathcal{P}^{(l)}_1 $ can be equivalently reformulated as
\begin{align} 
	\widehat{\mathcal{P}}^{(l)}_1: &  ~~ \underset{ \boldsymbol{\phi}_l, \boldsymbol{\gamma}_l, \tau_l }{\mathrm{minimize}}
	& & 
	\tau_{l}  \nonumber
	\\
	& ~~~~~~ \mathrm{s.t.} & \mathrm{C_{1}}: ~ &  \frac{\delta^2 K^2}{\sigma^2} \left| \mathbf{h}_l^\mathrm{H} \mathrm{diag} ( \mathbf{g}^{*} ) e^{\mathrm{j} \boldsymbol{\phi}_l} \right|^2 \geq \Gamma_{\mathrm{th},l}, \nonumber	
	\\
	& & \mathrm{C_2}: ~ & \phi_{l,n} \in [0, 2\pi), \forall n \in \mathcal{N}, \nonumber
	\\
	& & \mathrm{C_3}: ~ & \phi_{l,n} - \phi_{l-1,n} = \gamma_{l,n}, \forall n \in \mathcal{N}, \nonumber
	\\
	& & \mathrm{C_4}: ~ & \tau_l \geq 0, \nonumber
	\\
	& & \mathrm{C_5}: ~ & \max_{n \in \mathcal{N} } \hat{f} \left( \gamma_{l,n} \right) \leq \tau_l, \nonumber
\end{align}
where $ \boldsymbol{\gamma}_{l} = [ \gamma_{l,1}, \dots, \gamma_{l,N} ] $. By leveraging (\ref{eq:piecewise-approximation}), constraint $ \mathrm{C_5} $ can be recast as $ \max_{n \in \mathcal{N} } \max_{i \in \mathcal{I} } \hat{f}_i \left( \gamma_{l,n} \right) \leq \tau_l $, which is equivalent to
\begin{equation} \nonumber
	\mathrm{C_6}: \hat{f}_i \left( \gamma_{l,n} \right) \leq \tau_{l}, \forall n \in \mathcal{N}, i \in \mathcal{I}. \nonumber
\end{equation}

Hence, problem $ \widehat{\mathcal{P}}^{(l)}_1 $ is recast as
\begin{align} 
	\widetilde{\mathcal{P}}^{(l)}_1: \underset{ \boldsymbol{\phi}_l, \boldsymbol{\gamma}_l, \tau_l }{\mathrm{minimize}} ~~~~
	\tau_{l} ~~ \mathrm{s.t.} ~~ \mathrm{C_{1}}, \mathrm{C_2}, \mathrm{C_3}, \mathrm{C_4}, \mathrm{C_6}. \nonumber
\end{align} 

\textsc{Remark:} \emph{The objective of $ \widetilde{\mathcal{P}}^{(l)}_1 $ and constraints $ \mathrm{C}_{2} $, $ \mathrm{C}_{3} $, $ \mathrm{C}_{4} $, and $ \mathrm{C}_{6} $ are convex whereas constraint $ \mathrm{C}_{1} $ is nonconvex, which is addressed in Section \ref{subsection_transforming_phase_selection_multiple_linear_constraints} and Section \ref{subsection_convexifying_nonconvex_constraints}.}

\subsubsection{Transforming phase selection into multiple linear constraints} \label{subsection_transforming_phase_selection_multiple_linear_constraints}

The term $ e^{\mathrm{j} \boldsymbol{\phi}_l} $ in constraint $ \mathrm{C_{1}} $ makes direct optimization over the decision variables $ \phi_{l,n} $ challenging, as these appear within complex exponential functions. To address this, rather than optimizing over $ \phi_{l,n} $, we optimize over $ e^{j\phi_{l,n}} $, which dwell in the unit circle. To this end, we introduce new decision variables $ \mathbf{x}_l = \left[ x_{l,1} \dots, x_{l,N} \right]^\mathrm{T} \in \mathcal{S}^{N \times 1} $, where  $ \mathcal{S} = \left\lbrace 1,  e^{\mathrm{j} \frac{2 \pi}{Q}}, \dots, e^{\mathrm{j} \frac{2 \pi (Q-1)}{Q}} \right\rbrace $ is the set of admissible phases uniformly spaced in the unit circle\footnote{For LC-RISs, high phase-shift resolution can be readily achieved. This capability is represented by a large value of $ Q $, corresponding to $ \log_2(Q) $ quantization bits. In contrast, semiconductor-based RISs are typically restricted to coarse quantization levels, often limited to $1$- or $2$-bit resolution, since achieving higher resolution is prohibitively expensive for such technologies.}. The cardinality of this set is $ Q = \left| \mathcal{S} \right| $. By substituting $ e^{\mathrm{j} \boldsymbol{\phi}_l} $ with $ \mathbf{x}_l $, we eliminate the nonlinearity associated with optimizing $ \boldsymbol{\phi}_l $ in $ \mathrm{C}_{1} $, while introducing modifications to constraints $ \mathrm{C}_{2} $ and $ \mathrm{C}_{3} $. The original constraints $ \mathrm{C}_{1} $, $ \mathrm{C}_{2} $, and $ \mathrm{C}_{3} $ are then reformulated as $ \mathrm{C}_{7} $, $ \mathrm{C}_{8} $, $ \mathrm{C}_{9} $, $ \mathrm{C}_{10} $, and $ \mathrm{C}_{11} $, as shown in below
\begin{equation} \nonumber
	\mathrm{C_{1}}, \mathrm{C_{2}}, \mathrm{C_{3}} \Leftrightarrow
		\begin{cases}
			   	\mathrm{C_7}: \frac{\delta^2 K^2}{\sigma^2} \left| \mathbf{h}_l^\mathrm{H} \mathrm{diag} ( \mathbf{g}^{*} ) \mathbf{x}_l \right|^2 \geq \Gamma_{\mathrm{th},l}, 
			   	\\	
			   	\mathrm{C_{8}}: z_{l,n,q} \in \left\lbrace 0, 1 \right\rbrace, \forall n \in \mathcal{N}, q \in \mathcal{Q}, \nonumber
                \\
			   	\mathrm{C_9}: \mathbf{1}^\mathrm{T} \mathbf{z}_{l,n} = 1, \forall  \in \mathcal{N},
			   	\\	
			   	\mathrm{C_{10}}: x_{l,n} = \mathbf{s}^\mathrm{T} \mathbf{z}_{l,n}, \forall n \in \mathcal{N},
			   	\\	
			   	\mathrm{C_{11}}: \phase{x_{l,n}} - \phase{x_{l-1,n}} = \gamma_{l,n}, \forall n \in \mathcal{N}. \nonumber
		\end{cases}
\end{equation}

Here, $ \mathbf{s} \in \mathbb{C}^{Q \times 1} $ collects all the admissible phase values, i.e., $ \mathbf{s} = \big[ 1, e^{\mathrm{j} \frac{2 \pi}{Q}}, \dots, e^{\mathrm{j} \frac{2 \pi (Q-1)}{Q}} \big]^\mathrm{T} $. Constraint $ \mathrm{C}_{7} $ results directly from substituting $ e^{\mathrm{j} \boldsymbol{\phi}_l} $ with $ \mathbf{x}_l $. Constraint $ \mathrm{C}_{8} $ enforces all elements $ z_{l,n,q} $, constituents of vector $\mathbf{z}_{l,n} $, to be binary, where $ \mathcal{Q} = \left\lbrace 1, \dots, Q \right\rbrace $. Constraint $ \mathrm{C_9} $ ensures that exactly one element of $ \mathbf{z}_{l,n} $ equals $ 1 $. Consequently, $ \mathrm{C}_{10} $ guarantees that each $ x_{l,n} $ assumes one of the admissible values in $ \mathcal{S} $. The operator $ \phase{x_{l,n}} $ in constraint $ \mathrm{C}_{11} $ extracts the phase of $ x_{l,n} $ (in degrees), making $ \mathrm{C}_{11} $ equivalent to $ \mathrm{C}_{3} $.

The phase operator, however, is inherently nonlinear, as it is defined by $ \phase{x_{l,n}} = \arctan{\left( \frac{\mathsf{Im} \left( x_{l,n} \right) }{\mathsf{Re} \left( x_{l,n} \right) }\right) } $, which is highly nonconvex and difficult to handle analytically. Traditional methods approximate this nonlinearity by linearizing $ \arctan \left( x_{l,n} \right) $, but such processes introduce new challenges, most notably, sensitivity to initialization and the need for iterative refinement. Taking advantage that the phase shifts were discretized in the previous step, we propose the following strategy. We introduce vector $ \mathbf{r} $ defined as $ \mathbf{r} = \phase{\mathbf{s}} = \left[ 0, 360/Q, \dots, 360 (Q-1)/Q \right]  $, which stores the phase values (in degrees) corresponding to the elements of $ \mathcal{S} $. By ensuring  one-to-one correspondence between the elements of $ \mathbf{r} $ and $ \mathbf{s} $, i.e., $ r_q = \phase{s_q} $, we can express $ \mathrm{C_{11}} $ equivalently as
\begin{align*} \nonumber
	\mathrm{C_{12}}: \mathbf{r}^\mathrm{T} \mathbf{z}_{l,n} - \mathbf{r}^\mathrm{T} \mathbf{z}_{l-1,n} = \gamma_{l,n}, \forall n \in \mathcal{N}. 
\end{align*}

By applying the above transformations, $ \widetilde{\mathcal{P}}^{(l)}_1 $ reduces to $ \widebar{\mathcal{P}}^{(l)}_1 $
\begin{align} 
	\widebar{\mathcal{P}}^{(l)}_1: \underset{ \mathbf{x}_l, \mathbf{Z}_l, \boldsymbol{\gamma}_l, \tau_l }{\mathrm{minimize}} ~~~~
	\tau_{l} ~~ \mathrm{s.t.} ~~ \mathrm{C_4}, \mathrm{C_6}, \mathrm{C_7}, \mathrm{C_8}, \mathrm{C_9}, \mathrm{C_{10}}, \mathrm{C_{12}}, \nonumber
\end{align}
where $ \mathbf{Z}_l = \left[ \mathbf{z}_{l,1}, \dots, \mathbf{z}_{l,N} \right] \in \left\lbrace 0, 1 \right\rbrace^{Q \times N} $.

\textsc{Remark:} \emph{In $ \widebar{\mathcal{P}}^{(l)}_1 $, all constraints except $ \mathrm{C_{7}} $ are convex. In the following, we deal with the convexification of $ \mathrm{C_{7}} $.}

\subsubsection{Convexifying nonconvex constraints} \label{subsection_convexifying_nonconvex_constraints}

One possibility to deal with nonconvex constraint $ \mathrm{C_{7}} $ is to linearize it with respect to $ \mathbf{x}_l $ and iteratively refine the solution. However, we adopt a different strategy whereby we replace it by an inner convex approximation of $ \mathrm{C_{7}} $, relying on \textbf{Lemma \ref{lemma_3}}.
\begin{lemma} \label{lemma_3}
	For any complex variable $y \in \mathbb{C}$ and non-negative scalar $a \in \mathbb{R}_{+}$, the convex constraint $\mathsf{Re}\left\lbrace y \right\rbrace \geq a$ implies the nonconvex constraint $ |y| \geq a $. This provides a convex sufficient condition for the original nonconvex requirement, as $\left\lbrace y \in \mathbb{C} \mid \mathsf{Re} \left\lbrace y \right\rbrace  \geq a\right\rbrace \subseteq \left\lbrace y \in \mathbb{C} \mid |y| \geq a \right\rbrace $.
\end{lemma}

\begin{proof}
	Please, refer to Appendix \ref{app:proof-lemma-2}.
\end{proof}

Constraint $ \mathrm{C_{7}} $ is equivalent to $ \frac{\delta K}{\sigma} \left| \mathbf{h}_l^\mathrm{H} \mathrm{diag} ( \mathbf{g}^{*} )  \mathbf{x}_l \right|\geq \sqrt{\Gamma_{\mathrm{th},l}} $ after taking the square root of both sides. By invoking \textbf{Lemma \ref{lemma_3}}, constraint $ \mathrm{C_{7}} $ is satisfied if
\begin{align*}
	\mathrm{C_{13}}: \mathsf{Re} \left\lbrace \mathbf{h}_l^\mathrm{H} \mathrm{diag} ( \mathbf{g}^{*} ) \mathbf{x}_l \right\rbrace \geq \sigma\sqrt{\Gamma_{\mathrm{th},l}} / \delta K,
\end{align*}
holds. Consequently, $ \widebar{\mathcal{P}}^{(l)}_1 $ is finally transformed into
\begin{align*} 
	\check{\mathcal{P}}^{(l)}_1: & ~~ \underset{ \mathbf{x}_l, \mathbf{Z}_l, \boldsymbol{\gamma}_l, \tau_l }{\mathrm{minimize}}
	& & 
	\tau_{l}
	\\
	& ~~~~~~ \mathrm{s.t.} & \mathrm{C}_{4}: ~ &  \tau_l \geq 0,
	\\
	& & \mathrm{C}_{6}: ~ &  \hat{f}_i \left( \mathbf{r}^\mathrm{T} \mathbf{z}_{l,n} - \mathbf{r}^\mathrm{T} \mathbf{z}_{l-1,n} \right) \leq \tau_{l}, 
	\\
	& & & ~~~~~~~~~~~~~~~~~~~~ \forall n \in \mathcal{N}, i \in \mathcal{I}, \nonumber
	\\
	& & \mathrm{C}_{8}: ~ &  z_{l,n,q} \in \left\lbrace 0, 1 \right\rbrace, \forall n \in \mathcal{N}, q \in \mathcal{Q},  
	\\
	& & \mathrm{C}_{9}: ~ &  \mathbf{1}^\mathrm{T} \mathbf{z}_{l,n} = 1, \forall n \in \mathcal{N},  
	\\
	& & \mathrm{C}_{10}: ~ &  x_{l,n} = \mathbf{s}^\mathrm{T} \mathbf{z}_{l,n}, \forall n \in \mathcal{N},
	\\
	& & \mathrm{C}_{12}: ~ & \mathbf{r}^\mathrm{T} \mathbf{z}_{l,n} - \mathbf{r}^\mathrm{T} \mathbf{z}_{l-1,n} = \gamma_{l,n}, \forall n \in \mathcal{N},
	\\
	& & \mathrm{C}_{13}: ~ & \mathsf{Re} \left\lbrace \mathbf{h}_l^\mathrm{H} \mathrm{diag} ( \mathbf{g}^{*} ) \mathbf{Z}_l^\mathrm{T} \mathbf{s} \right\rbrace \geq \frac{\sigma\sqrt{\Gamma_{\mathrm{th},l}}}{\delta K}. 
\end{align*}

Specifically, problem $ \check{\mathcal{P}}^{(l)}_1 $ is a MILP that guarantees a globally optimal solution. Although its theoretical complexity increases exponentially with $ Q $, it can be solved efficiently in practice by leveraging mature optimization solvers such as MOSEK. In contrast to a brute-force approach that enumerates all possible phase-shift configurations with resolution $ Q $, requiring a computational complexity of $ \mathcal{O}(Q^N)$, the proposed formulation attains the optimal solution at only a fraction of that cost, achieving on average between $0.7\%$ to $5.3\%$ of the exhaustive search runtime in our simulations.

\textsc{Remark:} \emph{Since the feasible region defined by $ \mathrm{C}_{13} $ is a subset of that induced by $ \mathrm{C}_{7} $, any solution to $\check{\mathcal{P}}^{(l)}_1 $ remains feasible for the original problem $\widehat{\mathcal{P}}^{(l)}_1$. However, because $\check{\mathcal{P}}^{(l)}_1 $ operates over a reduced feasible space, the two problems are not guaranteed to yield identical optimal solutions.}

\subsection{\methodB}
\label{section_methodB}

Following the same procedure described in Section~\ref{section_methodA}, we transform problem $ {\mathcal{P}}_2 $ into
\resizebox{1.01\columnwidth}{!}{
  \begin{minipage}{\columnwidth}
  \begin{align*}
	\check{\mathcal{P}}_2: & ~~ \underset{\substack{ \mathbf{x}_1, \dots, \mathbf{x}_L \\ \mathbf{Z}_1, \dots, \mathbf{Z}_L, \\ \boldsymbol{\gamma}_1, \dots, \boldsymbol{\gamma}_L \\ \tau_1, \dots, \tau_L }}{\mathrm{minimize}}
	& & 
	\sum_{l \in \mathcal{L}} \tau_{l}
	\\
	& ~~~~~~ \mathrm{s.t.} & \mathrm{D}_{1}: ~ &  \tau_l \geq 0, \forall l \in \mathcal{L},
	\\
	& & \mathrm{D}_{2}: ~ &  \hat{f}_i \left( \mathbf{r}^\mathrm{T} \mathbf{z}_{l,n} - \mathbf{r}^\mathrm{T} \mathbf{z}_{l-1,n} \right) \leq \tau_{l}, 
	\\
	& & & ~~~~~~~~~~~~~~~~~~~~ \forall l \in \mathcal{L}, n \in \mathcal{N}, i \in \mathcal{I}, \nonumber
	\\
	& & \mathrm{D}_{3}: ~ &  z_{l,n,q} \in \left\lbrace 0, 1 \right\rbrace, \forall l \in \mathcal{L}, n \in \mathcal{N}, q \in \mathcal{Q},  
	\\
	& & \mathrm{D}_{4}: ~ &  \mathbf{1}^\mathrm{T} \mathbf{z}_{l,n} = 1, \forall l \in \mathcal{L}, n \in \mathcal{N},  
	\\
	& & \mathrm{D}_{5}: ~ &  x_{l,n} = \mathbf{s}^\mathrm{T} \mathbf{z}_{l,n}, \forall l \in \mathcal{L}, n \in \mathcal{N},
	\\
	& & \mathrm{D}_{6}: ~ & \mathbf{r}^\mathrm{T} \mathbf{z}_{l,n} - \mathbf{r}^\mathrm{T} \mathbf{z}_{l-1,n} = \gamma_{l,n}, \forall l \in \mathcal{L}, n \in \mathcal{N},
	\\
	& & \mathrm{D}_{7}: ~ & \mathsf{Re} \left\lbrace \mathbf{h}_l^\mathrm{H} \mathrm{diag} ( \mathbf{g}^{*} ) \mathbf{Z}_l^\mathrm{T} \mathbf{s} \right\rbrace \geq \frac{\sigma\sqrt{\Gamma_{\mathrm{th},l}}}{\delta K}, \forall l \in \mathcal{L}. 
\end{align*}
  \end{minipage}
}

Problem $ \check{\mathcal{P}}_2 $ is also a MILP, which can be efficiently solved using standard commercial solvers. In comparison to $ \check{\mathcal{P}}_1^{(l)} $, the number of decision variables in $ \check{\mathcal{P}}_2 $ increases by a factor of $ L $, resulting in a higher computational complexity on the order of $ \mathcal{O}(Q^{L N}) $. Nevertheless, as shown in Section~\ref{section_evalution}, the joint optimization of all phase-shift configurations in $ \check{\mathcal{P}}_2 $ achieves further reductions in the overall reconfiguration time.

\section{Simulation results}   
\label{section_evalution}

This section first evaluates the performance of \sysname, analyzing the impact of phase resolution, SNR threshold, and CSI uncertainty on the reconfiguration time. Subsequently, we compare \sysname with a conventional baseline, referred to as \legacy, which is agnostic to LC's response time.

\subsection{Parameter settings and baseline} \label{sec:parameters-and-baseline}

Unless stated otherwise, we consider $ \delta = 1 $, $ P_\mathrm{BS} = 20 $~W, $ G_\mathrm{BS} = 10 $~dBi, $ G_\mathrm{MT} = 3 $~dBi, $ \sigma^2 = -110 $~dB, $ \Lambda = \Lambda_l = 100 $, and $ \Gamma_{\mathrm{th},l} = \Gamma_{\mathrm{th}} $. The BS is positioned $10$~m from the LC-RIS with fixed angles $ \beta_\mathrm{r}^{(l)} = 50^\circ $ and $ \theta_\mathrm{r}^{(l)} = 90^\circ $. For each MT, its distance from the LC-RIS is drawn uniformly from $[8, 12]$~m, with elevation angle $ \theta_\mathrm{t}^{(l)} = 90^\circ $ and azimuth angle $ \beta_\mathrm{t}^{(l)} $ sampled uniformly from $[95^\circ, 175^\circ]$. To match our mmWave LC-RIS prototype, we set the carrier frequency to $ f_\mathrm{c} = 61 $~GHz and configure the RIS with $ N_\mathrm{x} = 12 $ and $ N_\mathrm{z} = 10 $. Path loss is modeled using the 3GPP UMa model \cite{3gpp:38.901}. All results are averaged over $200$ independent realizations.

\subsection{Impact of quantization and SNR threshold}

This scenario evaluates the reconfiguration time across different phase quantization levels $ Q $, with the goal of identifying a practical quantization level that offers a favorable tradeoff between performance and computational complexity. In parallel, we examine the impact of the SNR threshold  $\Gamma_{\mathrm{th}}$ on the reconfiguration time. For this study, we employ \methodA, as it is more sensitive to quantization and phase selection than \methodB, which partially mitigates these effects through joint optimization across multiple phase-shift configurations.

\begin{figure}[!t]
	\centering
	\includegraphics[]{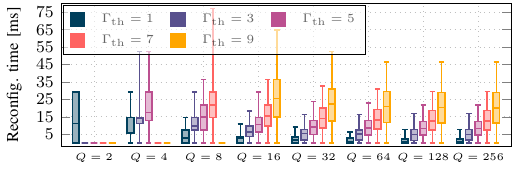} 
	\caption{Impact of quantization and SNR threshold.}
	\label{fig:scenario-1}
	\vspace{-2mm}
\end{figure}
Fig.~\ref{fig:scenario-1} presents the reconfiguration times using a box-and-whisker plot, where the horizontal markers (from top to bottom) denote the \textbf{maximum}, \textbf{75th percentile}, \textbf{mean}, \textbf{25th percentile}, and \textbf{minimum}. Increasing the phase resolution consistently decreases the reconfiguration time. Specifically, the reduction is substantial for $ Q \leq 32 $, while it nearly saturates for $Q \geq 64$. The figure also shows that reconfiguration time grows as the SNR threshold $ \Gamma_{\mathrm{th}} $ becomes more stringent. Specifically, for low $ \Gamma_{\mathrm{th}} $, many phase-shift configurations satisfy the constraint with only minor adjustments, resulting in short reconfiguration times. As $ \Gamma_{\mathrm{th}} $ increases, the feasible space contracts, often demanding larger phase transitions, which leads to longer reconfiguration times.

We note that empty regions without a box-and-whisker plot correspond to infeasible optimization instances, where the limited phase resolution prevents satisfying the required SNR threshold. Overall, $ Q = 64 $ yields an average error below $ 4\% $ relative to $ Q = 256 $, while $ Q = 128 $ reduces this error to below $ 1\% $. Since the goal is to balance performance and computational complexity, $ Q = 64 $ offers an effective and practical choice for subsequent simulations.

\subsection{Impact of imperfect CSI}

This scenario evaluates the reconfiguration time under different levels of CSI uncertainty. We adopt a sampling-based approach in which constraints $ \mathrm{C}_{13} $ and $ \mathrm{D}_{7} $ are enforced over $ U $ perturbed channel realizations. Each realization is generated by adding an error vector to the nominal channel. Formally, the constraints take the following form
\begin{align*}
	\mathsf{Re} \left\lbrace \left( \mathbf{h}_l + \mathbf{e}_{l,u} \right)^\mathrm{H} \mathrm{diag} ( \left( \mathbf{g} + \mathbf{e}_u \right)^{*} ) \mathbf{Z}_l^\mathrm{T} \mathbf{s} \right\rbrace \geq \frac{\sigma\sqrt{\Gamma_{\mathrm{th},l}}}{\delta K}, \forall u \in \mathcal{U}, 
\end{align*}
where the perturbation vectors $ \mathbf{e}_{l,u} $ and $ \mathbf{e}_u $ satisfy $ \left\| \mathbf{e}_{l,u} \right\|_2 \leq \epsilon_l $ and  $ \left\| \mathbf{e}_u \right\|_2 \leq \epsilon $, respectively, and $ \epsilon_l $ and $ \epsilon $ denote the channel error magnitudes. The index set of samples is $ \mathcal{U} = \left\lbrace 1, \dots, U \right\rbrace $.

\begin{figure}[!t]
 	\begin{subfigure}[b]{1\columnwidth}
		\begin{center}
		\includegraphics[]{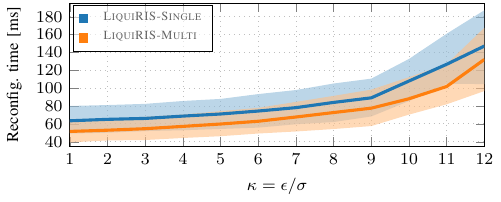}
		\vspace{-1mm} 
		\caption{Tradeoff between reconfiguration time and CSI imperfection}
		\vspace{2mm}
	 	\label{fig:scenario-2a}
		\end{center}
 	\end{subfigure}
 	\begin{subfigure}[b]{0.48\columnwidth}
		\begin{center}
		\includegraphics[]{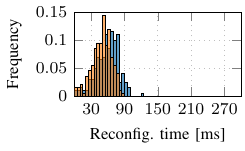} 
		\vspace{-2mm}
		\caption{Histogram $|$ $ \epsilon = 0 $}
		\vspace{2mm}
	 	\label{fig:scenario-2b}
		\end{center}
 	\end{subfigure}
 	\begin{subfigure}[b]{0.48\columnwidth}
		\begin{center}
		\includegraphics[]{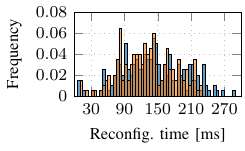} 
		\vspace{-2mm}
		\caption{Histogram $|$ $ \epsilon = 12 \sigma $}
		\vspace{2mm}
	 	\label{fig:scenario-2c}
		\end{center}
 	\end{subfigure}
 	\vspace{-3mm}
    \caption{Impact of imperfect CSI.}
 	\label{fig:scenario-2}
 	\vspace{-2mm}
\end{figure}
\begin{figure}[!t]
 	\begin{subfigure}[b]{0.48\columnwidth}
		\begin{center}
		\includegraphics[]{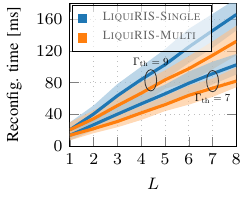}
		\vspace{-5mm} 
		\caption{\sysname $|$ $ \Gamma_{\mathrm{th}} = \left\lbrace 7, 9 \right\rbrace $}
		\vspace{2mm}
	 	\label{fig:scenario-3a}
		\end{center}
 	\end{subfigure}
 	\begin{subfigure}[b]{0.48\columnwidth}
		\begin{center}
		\includegraphics[]{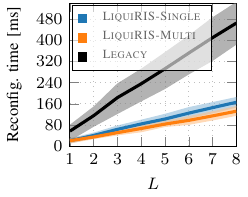} 
		\vspace{-5mm}
		\caption{All approaches $|$ $ \Gamma_{\mathrm{th}} = 9 $}
		\vspace{2mm}
	 	\label{fig:scenario-3b}
		\end{center}
 	\end{subfigure}
 	\begin{subfigure}[b]{0.48\columnwidth}
		\begin{center}
		\includegraphics[]{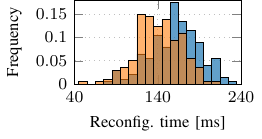} 
		\vspace{-5mm}
		\caption{Histogram $|$ $ L = 8 $ }
		\vspace{2mm}
	 	\label{fig:scenario-3c}
		\end{center}
 	\end{subfigure}
 	\begin{subfigure}[b]{0.48\columnwidth}
		\begin{center}
		\includegraphics[]{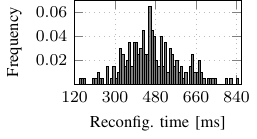} 
		\vspace{-5mm}
		\caption{Histogram $|$ $ L = 8 $ }
		\vspace{2mm}
	 	\label{fig:scenario-3d}
		\end{center}
 	\end{subfigure}
 	\vspace{-3mm}
    \caption{Impact of number of MTs.}
 	\label{fig:scenario-3}
 	\vspace{-5mm}
\end{figure}

Fig.~\ref{fig:scenario-2a} reports the reconfiguration time of \methodA and \methodB for $ L = 3 $, $ \Gamma_{\mathrm{th}} = 9 $, $ U = 1000 $, and $ \epsilon_l = \epsilon = \kappa \sigma $, where the perturbation level $ \epsilon $ is modeled proportional to the channel standard deviation $ \sigma $, and $ \kappa $ serves as a scaling factor. The solid curves show the average reconfiguration time, while the shaded regions indicate the $25$th-$75$th percentile range.

Importantly, the increase in reconfiguration time with CSI uncertainty is a direct consequence of the robust feasibility requirement across multiple channel realizations. Specifically, the LC-RIS response time is determined by the extent of phase change required when transitioning from one configuration to another. Under perfect CSI, the optimizer can select phase-shift configurations that satisfy the SNR constraint while remaining close to the previous configuration, thereby minimizing phase variations and response time. In contrast, under imperfect CSI, the design must satisfy the SNR constraint simultaneously for a set of perturbed channels. This significantly shrinks the feasible set of phase-shift configurations and often forces the optimizer to select solutions that are farther away from the previous configuration. As a result, the required phase transitions become larger on average, which directly increases the reconfiguration time through the phase-dependent response model $ \hat{f}(\cdot) $.

In particular, this effect becomes more pronounced as $ \kappa $ increases. When $ \kappa $ is small, the perturbed channels remain close to the nominal one, and the robust solution is similar to the nominal solution, resulting in only a minor increase in reconfiguration time. However, for larger $ \kappa $, the LC-RIS must accommodate a much wider range of channel variations, which leads to significantly larger phase adjustments and, consequently, longer reconfiguration times. We observe that this increase is approximately linear for $ \kappa \leq 6 $, while it grows rapidly for larger values of $ \kappa $, reflecting the increasingly restrictive feasible set.

Fig.~\ref{fig:scenario-2b} and Fig.~\ref{fig:scenario-2c} show the corresponding histograms for $ \epsilon = 0 $ and $ \epsilon = 12\sigma $, respectively. As $ \epsilon $ increases, the distribution widens considerably. While most reconfiguration times lie between $0$ and $90$ ms for $ \epsilon = 0 $, they spread over the range $0$ to $290$ ms for $ \epsilon = 12\sigma $, highlighting the strong impact of robustness requirements on phase transitions and, consequently, on reconfiguration time.

\subsection{Comparison with baseline}

This section first evaluates the impact of the number of MTs and the SNR threshold on the reconfiguration time of \methodA and \methodB, and then compares both methods with the \legacy baseline to contextualize their performance gains. The \legacy baseline designs the phase-shift configuration for each MT to maximize its SNR without accounting for the reconfiguration time overhead induced by phase transitions. It is formulated as
\begin{align} 
	\check{\mathcal{P}}_3: & ~~ \underset{\substack{ \mathbf{Z}_1, \dots, \mathbf{Z}_L, \\ \alpha_1, \dots, \alpha_L }}{\mathrm{maximize}}
	& & 
	\sum_{l \in \mathcal{L}} \alpha_{l}  \nonumber
	\\
	& & \mathrm{E}_{1}: ~ & z_{l,n,q} \in \left\lbrace 0, 1 \right\rbrace, \forall l \in \mathcal{L}, n \in \mathcal{N}, q \in \mathcal{Q}, \nonumber
	\\
	& & \mathrm{E}_{2}: ~ & \mathbf{1}^\mathrm{T} \mathbf{z}_{l,n} = 1, \forall l \in \mathcal{L}, n \in \mathcal{N}, \nonumber
	\\
	& & \mathrm{E}_{3}: ~ & \frac{\delta K}{\sigma} \mathsf{Re} \left\lbrace \mathbf{h}_l^\mathrm{H} \mathrm{diag} ( \mathbf{g}^{*} ) \mathbf{Z}_l^\mathrm{T} \mathbf{s} \right\rbrace \geq \alpha_l, \nonumber	
\end{align} 
where $ \alpha_l $ represents an auxiliary variable. Unlike $ \check{\mathcal{P}}_2$, in which the phase-shift configuration for each MT depends on the preceding configuration, $ \check{\mathcal{P}}_3 $ ignores this dependency. As a result, solving the problems sequentially or jointly leads to identical solutions.

Fig.~\ref{fig:scenario-3a} shows the reconfiguration time of \methodA and \methodB as the number of MTs increases, for two SNR thresholds $ \Gamma_{\mathrm{th}} = \left\lbrace 7, 9 \right\rbrace $. Solid lines denote the average and shaded regions span the $25$th-$75$th percentiles. A key observation is that the performance gap between \methodA and \methodB widens with $ L $. When $ \Gamma_\mathrm{th} = 9 $, \methodB is approximately $6$ms faster than \methodA at $ L = 2 $, and nearly $34$ms faster at $ L = 8 $. For $ \Gamma_\mathrm{th} = 7 $, a similar trend is observed, although the difference between the two approaches is less pronounced. These results indicate that \methodA is well suited for scenarios with a small number of MTs, as it incurs lower computational complexity and the additional gains of \methodB remain limited. When the number of MTs becomes moderate, for example $ L \geq 4 $, \methodB becomes the preferred option because the reconfiguration time savings are substantial enough to justify its higher computational cost.

Fig.~\ref{fig:scenario-3b} shows the performance of \legacy, \methodA, and \methodB for $ \Gamma_\mathrm{th} = 9 $, highlighting a pronounced gap in reconfiguration time. Relative to \legacy, \methodB reduces the reconfiguration time by $ 71.61\% $ on average, while \methodA achieves a $ 64.36\% $ reduction.
Fig.~\ref{fig:scenario-3c} and Fig.~\ref{fig:scenario-3d} depict the histograms for $ L = 8 $ and $ \Gamma_\mathrm{th} = 9 $. \methodB exhibits a much tighter concentration around shorter reconfiguration times, whereas \methodA is shifted to the right, indicating noticeably longer delays. We also observe that \legacy spans a broad interval from $120$ to $840$ ms, which explains its significantly higher reconfiguration time.

\section{Experimental results}

This section presents real-world results obtained using our LC-RIS prototype, which consists of  $ 120 $ unit cells, as mentioned in Section~\ref{sec:parameters-and-baseline}.\footnote{Additional details of the employed LC-RIS prototype are provided in Appendix~\ref{app:characteristics-prototype}.} The measurement setup, illustrated in Fig.~\ref{figure_Measurement_Setup}, employs two mmWave V-band horn antennas, denoted as \textsc{Antenna~1} and \textsc{Antenna~2}, positioned approximately $ 1 $m and $ 60 $ cm from the LC-RIS, respectively.
\begin{figure}[!h]
    \centering
    \includegraphics[width=0.45\textwidth]{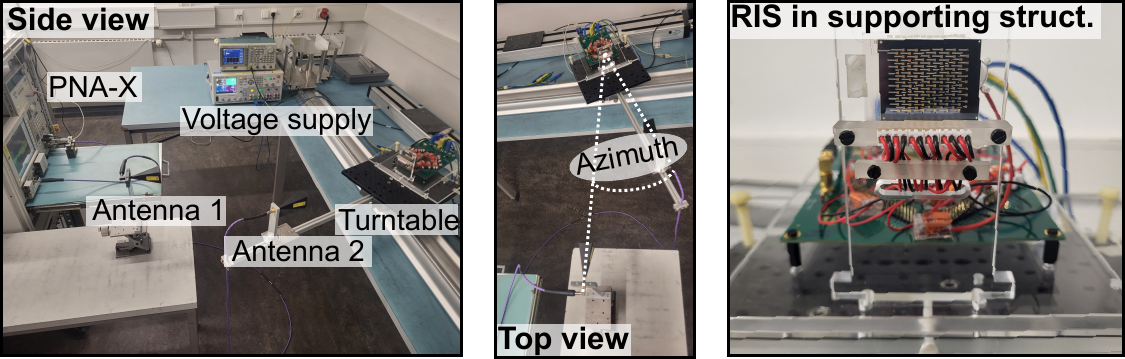} 
    \caption{Bistatic measurement setup for \sysname.}
    \label{figure_Measurement_Setup}
\end{figure}

\textsc{Antenna~2} is mechanically coupled to the turntable on which the LC-RIS is mounted, ensuring co-rotation during angular sweeps. Both antennas are connected to a Keysight PNA-X N5247A network analyzer. The LC-RIS is driven by a DAC60096 EVM that supplies a $1$ kHz square-wave biasing signal. To isolate the LC-RIS reconfiguration time from environmental reflections, two reference measurements are acquired and time gating is applied.

To evaluate the performance of \sysname and compare it against the \legacy baseline, we focus on two aspects:

$\bullet$ \emph{Beampattern}. Since \sysname computes phase-shift configurations under stricter constraints than \legacy, we measure the resulting beampatterns to assess how these constraints translate into directional performance.

$\bullet$ \emph{Reconfiguration time}. As \sysname is explicitly designed to reduce the reconfiguration time between successive phase-shift configurations, we also measure the resulting reconfiguration time required to form the desired beampattern.

 \subsection{Performance of \methodA}
\label{ss_methodA}

Fig.~\ref{fig:scenario-4} evaluates the performance of \methodA relative to \legacy for the case $ L = 4 $ with azimuth angles $ \boldsymbol{\beta}_\mathrm{t} = \left[ 140^\circ, 120^\circ, 60^\circ, 40^\circ \right] $\footnote{ The default starting angle is $90^\circ$ (i.e., $ \beta_\mathrm{t} = 90^\circ $), and the initial phase vector for $ l = 0 $ is set to $ \mathbf{x}_0 = [1,\dots,1]^\mathrm{T} $. Since the number of MTs is small, we set $ Q = 256 $ to ensure high-precision phase quantization.}. As shown in Fig.~\ref{fig:scenario-4a}, \methodA produces beampatterns that closely match those of \legacy. In particular, $(i)$ the main-lobe gain achieved by \methodA is nearly identical to that of \legacy, and $(ii)$ the overall beampattern maintains high fidelity, confirming that \methodA does not introduce distortions. The corresponding reconfiguration times in Fig.~\ref{fig:scenario-4b} further highlight the advantages of \methodA. On average across all transitions, it reduces the reconfiguration time by $39.67\%$, with a maximum reduction of $65.58\%$ when transitioning from $120^\circ$ to $60^\circ$ and a minimum reduction of $25.14\%$ when transitioning from $60^\circ$ to $ 40^\circ $.

\begin{figure}[!t]
	\centering
	
	\begin{subfigure}{1\columnwidth}
	  \begin{center}
	  \vspace{1mm}
	  \includegraphics[height=5.8cm]{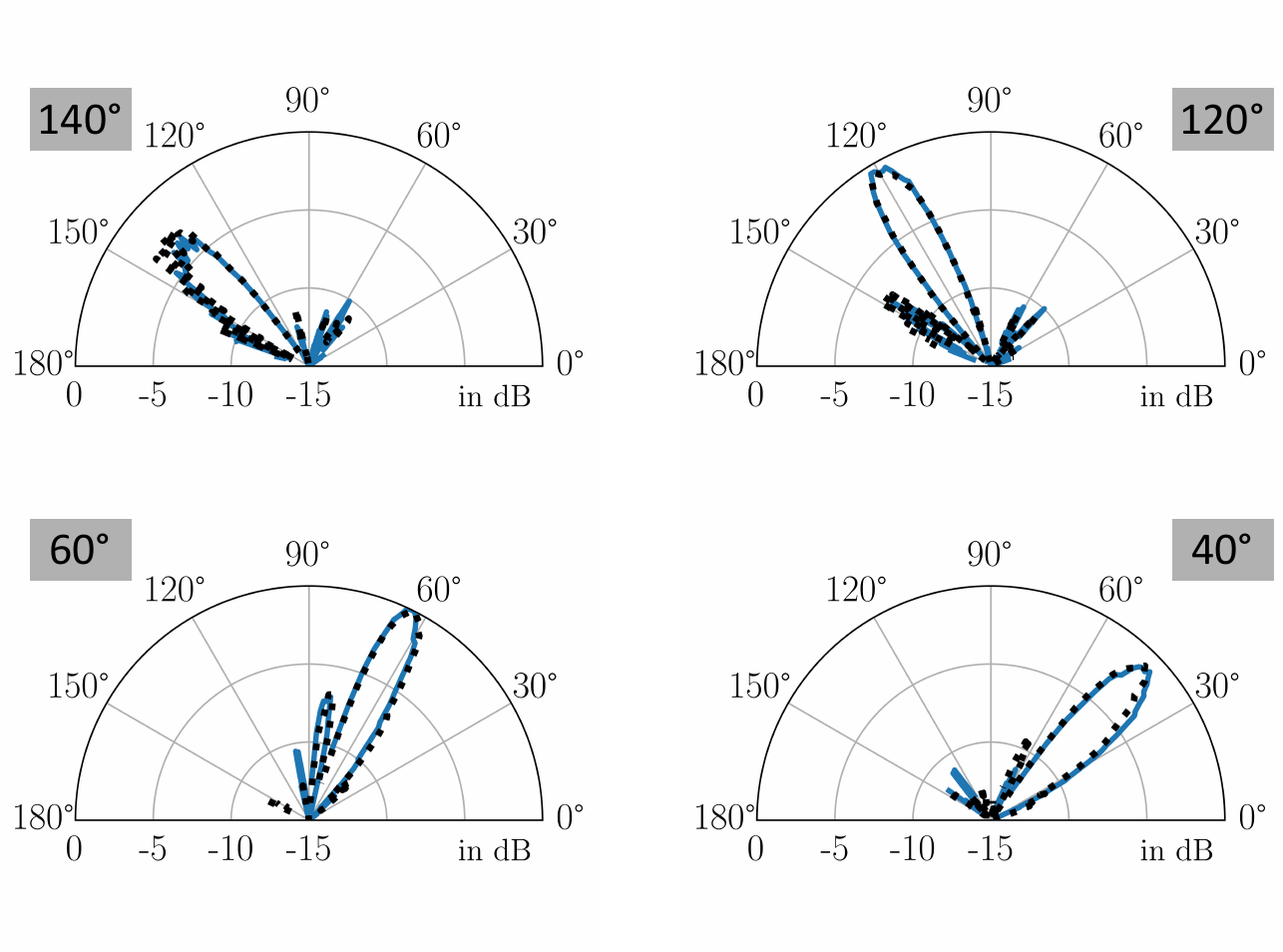} 
	  \caption{}
	  \label{fig:scenario-4a}
	  \vspace{3mm}
	  \end{center}
	\end{subfigure}
	
	\begin{subfigure}{1\columnwidth}
	    \begin{center}
	    \includegraphics[height=4.2cm]{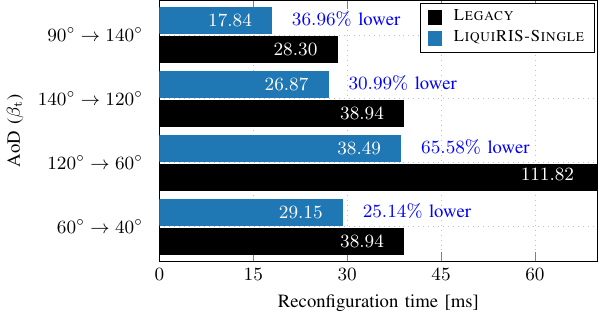}
	    \caption{}
	    \label{fig:scenario-4b}    
	    \end{center}
	\end{subfigure}
	\caption{Comparing \methodA and \legacy.}
	\label{fig:scenario-4}
	\vspace{-3mm}
\end{figure}
\begin{figure}[!t]
    
    \begin{subfigure}{1\columnwidth}
    	\begin{center}
        \includegraphics[height=5.8cm]{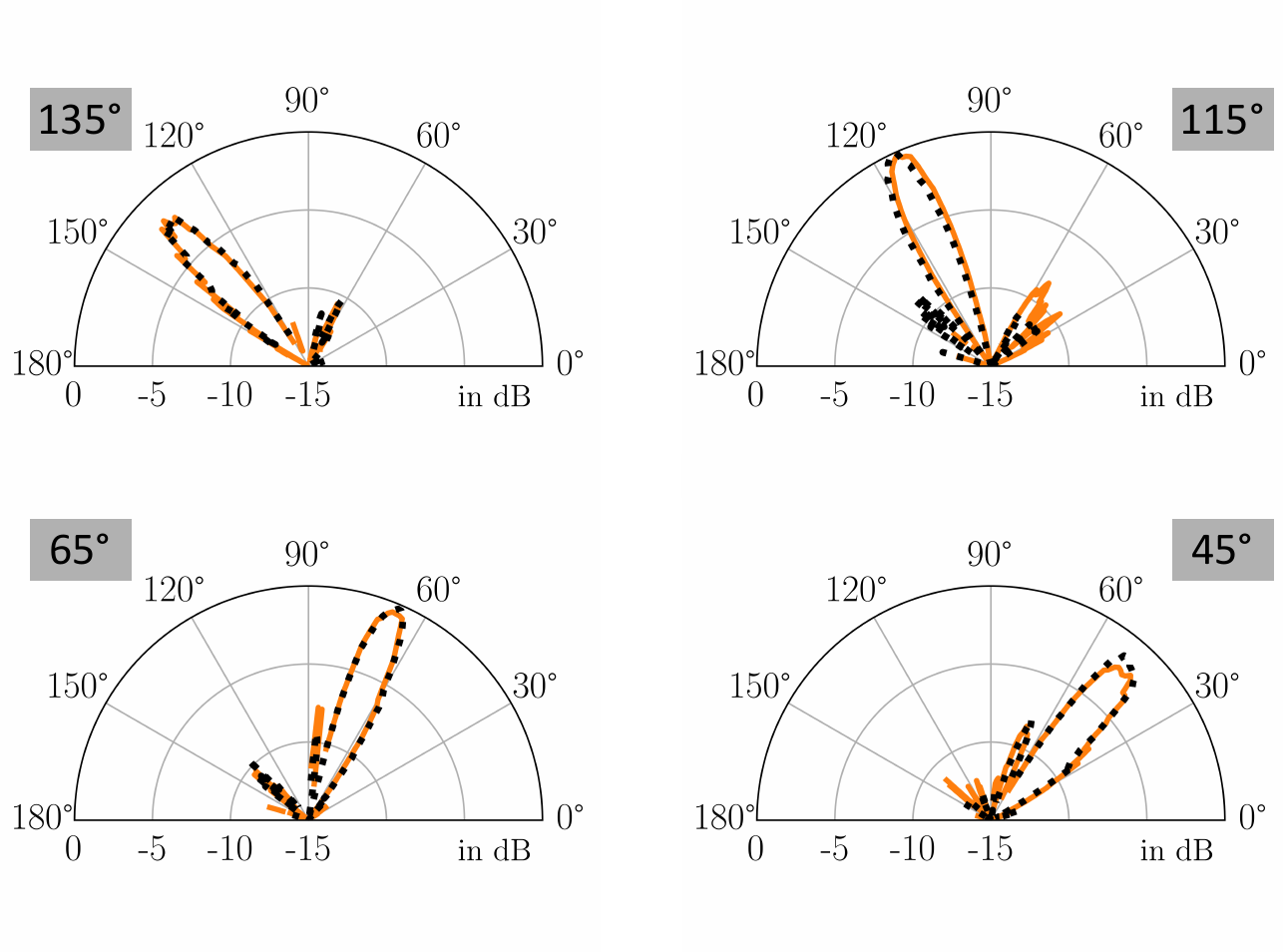} 
          \caption{}
        \label{fig:scenario-5a}    
        \vspace{3mm}
        \end{center}
    \end{subfigure}%

    \begin{subfigure}{1\columnwidth}
    \begin{center}
      \includegraphics[height=4.2cm]{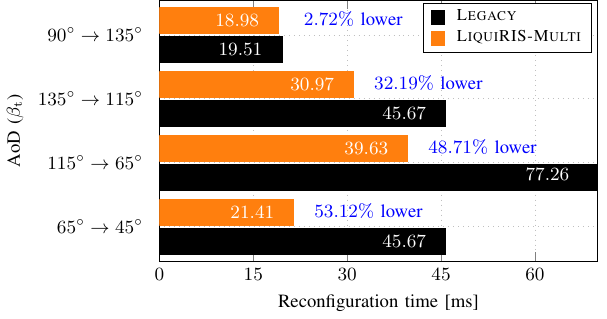}
        \caption{}
        \label{fig:scenario-5b}     
    \end{center}      
    \end{subfigure}
    
    \caption{Comparing \methodB and \legacy.}
    \label{fig:scenario-5}
    \vspace{-3mm}
\end{figure}

\subsection{Performance of \methodB}
\label{ss_methodb}

Fig.~\ref{fig:scenario-5} evaluates the performance of \methodB relative to \legacy. In this setting, \methodB jointly designs all phase-shift configurations for $ L = 4 $, using azimuth angles $ \boldsymbol{\beta}_\mathrm{t} = \left[ 135^\circ, 115^\circ, 65^\circ, 45^\circ \right] $. As shown in Fig.~\ref{fig:scenario-5a}, \methodB produces beampatterns that closely match those of \legacy, indicating that optimizing for reconfiguration time does not significantly degrade the beampattern quality. Meanwhile, Fig.~\ref{fig:scenario-5b} demonstrates that \methodB substantially reduces the reconfiguration time. Across the four transitions, \methodB yields an average reduction of $34.19\%$ compared to \legacy.

The reduction achieved by \methodB is lower than the $ 39.67\% $ reduction obtained by \methodA in the previous scenario. This difference, however, stems from the differing channel conditions, which are strongly characterized by $ \boldsymbol{\beta}_\mathrm{t} $. In the next scenario, we evaluate both approaches under identical settings to enable a direct comparison.

\begin{figure}[!t]
	\begin{subfigure}{1\columnwidth}
		\begin{center}
			\vspace{1mm}
			\includegraphics[height=5.8cm]{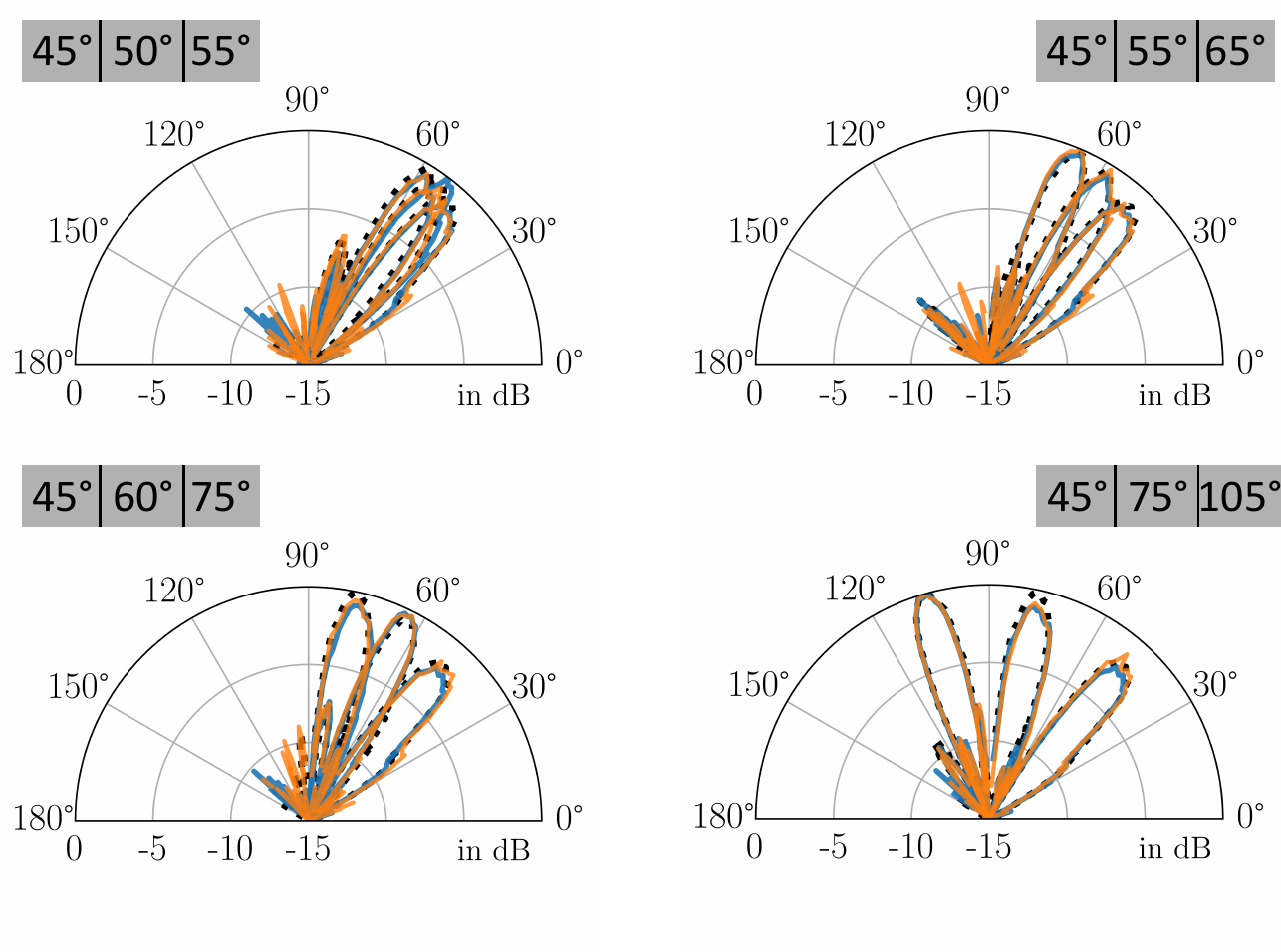} 
			\caption{}
			\label{fig:scenario-6a}   
			\vspace{3mm}
		\end{center} 
	\end{subfigure}%
	
	\begin{subfigure}{1\columnwidth}
		\begin{center}
			\includegraphics[height=6.5cm]{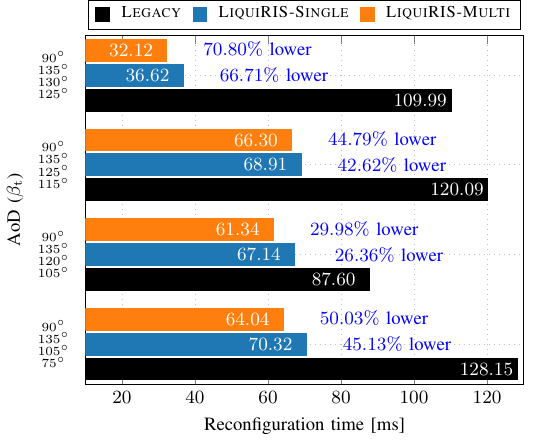}
			\caption{}
			\label{fig:scenario-6b}  
		\end{center}       
	\end{subfigure}
	\caption{Comparing all three approaches.}
	\label{fig:scenario-6}
	\vspace{-5mm}
\end{figure}

\subsection{Impact of angular separation} \label{ss_angular_separation}

Fig.~\ref{fig:scenario-6} compares all three approaches under varying angular separations between the served MTs. We consider $ L = 3 $, where the MTs are separated by $5^\circ$, $10^\circ$, $15^\circ$, and $30^\circ$. As shown in Fig.~\ref{fig:scenario-6a}, both \methodA and \methodB generate beampatterns that closely match those of \legacy across all angular separations. The reconfiguration times reported in Fig.~\ref{fig:scenario-6b} show that \sysname consistently outperforms \legacy, regardless of how closely spaced the MTs are. Moreover, \methodB achieves an average reconfiguration time that is $8.41\%$ lower than that of \methodA. In this setting, \methodB yields maximum and minimum reductions of $70.80\%$ and $29.98\%$, respectively, while \methodA achieves maximum and minimum reductions of $66.71\%$ and $26.36\%$.

Overall, across all transitions, \sysname reduces the total reconfiguration time to approximately half of that required by \legacy, representing a substantial improvement in responsiveness.

\section{Conclusion}

This paper advances the practical deployment of RISs by introducing a novel, physics-informed phase-shift optimization design for \gls{LC-RIS}. Our approach directly tackles the critical bottleneck of response time, which is a key challenge for large-scale adoption. Through comprehensive system-level simulations, we have evaluated the impact of real-world parameters, such as the number of MTs, SNR thresholds, phase-quantization levels, and imperfect CSI, on reconfiguration time. We also validated the practical feasibility of this design through real-world experiments with a functional hardware prototype. Our work demonstrates that \sysname provides a cost-effective solution for integrating RIS technology into future wireless systems, advancing the promise of large-scale programmable radio environments.

\begin{table*} [!t]
	\vspace{4mm}
	\begin{center}
		\begin{tabular}{|c >{\columncolor[gray]{0.9}}c c >{\columncolor[gray]{0.9}}c c >{\columncolor[gray]{0.9}}c c >{\columncolor[gray]{0.9}}c c >{\columncolor[gray]{0.9}}c c >{\columncolor[gray]{0.9}}c c >{\columncolor[gray]{0.9}}c c >{\columncolor[gray]{0.9}}c c|} 
		 \hline
		 $ w_{1} $ & $ w_{2} $ & $ w_{3} $ & $ w_{4} $ & $ w_{5} $ & $ w_{6} $ & $ w_{7} $ & $ w_{8} $ & $ w_{9} $ & $ w_{10} $ & $ w_{11} $ & $ w_{12} $ & $ w_{13} $ & $ w_{14} $ & $ w_{15} $ & $ w_{16} $ & $ w_{17} $ \\ 
		 \hline
		 $ 1990 $ & $ 690 $ & $ 410 $ & $ 230 $ & $ 170 $ & $ 120 $ & $ 80 $ & $ 40 $ & $ 0 $ & $ 20 $ & $ 30 $ & $ 50 $ & $ 90 $ & $ 130 $ & $ 250 $ & $ 450 $ & $ 1620 $ \\ 
		 \hline
		 \hline 
		 $ c_{1} $ & $ c_{2} $ & $ c_{3} $ & $ c_{4} $ & $ c_{5} $ & $ c_{6} $ & $ c_{7} $ & $ c_{8} $ & $ c_{9} $ & $ c_{10} $ & $ c_{11} $ & $ c_{12} $ & $ c_{13} $ & $ c_{14} $ & $ c_{15} $ & $ c_{16} $ & $ c_{17} $ \\
		 \hline
		 $ -360 $ & $ -359 $ & $ -358 $ & $ -356 $ & $ -352 $ & $ -342 $ & $ -320 $ & $ -247 $ & $ 0 $ & $ 320 $ & $ 336 $ & $ 346 $ & $ 352 $ & $ 354 $ & $ 356 $ & $ 358 $ & $ 360 $ \\
		 \hline
		\end{tabular}
		\caption{Breakpoints defining piecewise linear function $ \hat{f} $.}
		\label{table_samples}
	\end{center}
	\vspace{-4mm}
\end{table*}

\section*{Acknowledgment}

The research was in part funded by the Deutsche Forschungsgemeinschaft (DFG) within the B5G-Cell project (210487104) in SFB 1053 MAKI, and the C09 project within the TRR 196 Marie (287022738), in part by the HyRIS project (455077022), in part by the LOEWE initiative (Hesse, Germany) within the emergenCITY Centre under Grant LOEWE/1/12/519/03/05.001(0016)/72, and in part by the German Federal Ministry for Research, Technology and Space (BMFTR) under the program of ``Souverän. Digital. Vernetzt.'' joint project Open6GHub plus (Project-ID 16KIS2407).

\appendix

\subsection{Proof to Lemma 1} \label{app:proof-lemma-1}

        Since $ \hat{f} $ is convex, the following inequality holds true, for $ i \in \mathcal{I} $ and $ \theta \in \left[ 0,1 \right] $,
	\begin{align} \label{equation_convexity_definition}
		\hat{f} \left( \theta \phi' + \left( 1 - \theta \right) \phi'' \right) \leq \theta \hat{f} \left( \phi' \right) + \left( 1 - \theta \right) \hat{f} \left( \phi'' \right). 
	\end{align}
 
	Upon reordering the terms in (\ref{equation_convexity_definition}), we obtain 
	\begin{align} \label{equation_alternative_convexity_definition}
		\hat{f} \left( \phi' \right) \geq \hat{f} \left( \phi'' \right) + \frac{\hat{f} \left( \phi'' + \theta \left( \phi' - \phi'' \right) \right) - \hat{f} \left( \phi''\right) }{\theta}.
	\end{align}
	
	We assume that $ \phi' \in \Phi_i $ and $ \phi'' \in \Phi_j $, $ i, j \in \mathcal{I} $, and choose $ \theta $ such that $ \phi^\star = \phi'' + \theta \left( \phi' - \phi'' \right) \in \Phi_j $. From (\ref{equation_piecewise_linear_function}), we have
	\begin{align}
		& \hat{f} \left( \phi' \right) = \hat{f}_i \left( \phi' \right) = a_i \phi' + b_i, \nonumber \\
		& \hat{f} \left( \phi'' \right) = \hat{f}_j \left( \phi'' \right) = a_j \phi'' + b_j, \nonumber \\
		& \hat{f} \left( \phi^\star \right) = \hat{f}_j \left( \phi^\star \right) = a_j \phi^\star + b_j. \nonumber
	\end{align}
	
 	Upon replacing the above relations in (\ref{equation_alternative_convexity_definition}), it follows that 
	\begin{align}
		a_i \phi' + b_i & \geq a_j \phi' + b_j, \phi' \in \Phi_i, \forall j \in \mathcal{I}, \nonumber
	\end{align}
	which is equivalent to
	\begin{align}
		\hat{f}_i \left( \phi' \right) & \geq \hat{f}_j \left( \phi' \right), \phi' \in \Phi_i, \forall j \in \mathcal{I}, \nonumber \\
									   & = \max_{j \in \mathcal{I}} \hat{f}_j \left( \phi' \right) = \hat{f} \left( \phi' \right), \phi' \in \Phi_i, \nonumber 
	\end{align}	
	thus completing the proof.

\subsection{Parameters of piecewise linear function } \label{app:parameters-piecewise-function}

Table~\ref{table_samples} summarizes the parameters used to approximate the response-time function.

\subsection{Proof of Lemma 2} \label{app:proof-lemma-2}

By definition, $|y| = \sqrt{(\mathsf{Re} \left\lbrace y \right\rbrace )^2 + (\mathsf{Im} \left\lbrace y \right\rbrace )^2}$. It follows directly that $|y| \geq |\mathsf{Re} \left\lbrace y \right\rbrace |$. If $\mathsf{Re} \left\lbrace y \right\rbrace \geq a$, for $ a \geq 0 $, then $|y| \geq \mathsf{Re} \left\lbrace y \right\rbrace \geq a$. Thus, $ \mathsf{Re} \left\lbrace y \right\rbrace \geq a $ is a sufficient condition for $|y| \geq a$. The lower bound is tight when $\mathsf{Im} \left\lbrace y \right\rbrace = 0 $ and $ \mathsf{Re} \left\lbrace y \right\rbrace \geq 0 $, making $|y| = \mathsf{Re} \left\lbrace y \right\rbrace $.

\subsection{Characteristics of LC-RIS prototype} \label{app:characteristics-prototype}

We summarize herein the key characteristics of the LC-RIS prototype, which is described in detail in \cite{neuder_architecture_2023}. 
The prototype consists of 120 unit cells and is designed to operate at a central frequency of \SI{60}{\GHz}.
Based on the delay line architecture, a dedicated, LC-filled phase shifter is employed for the phase tuning in each unit cell.
The phase shifters are coupled to patch antenna elements for the reception and re-reradiation of electromagnetic waves.

The prototype features GT7-29001 from Merck Electronics KGaA, Darmstadt, Germany as the LC-mixture, AF32 glass from Schott and gold as a conductor.
The glass substrate is \SI{700}{\micro \m} thick and the LC layer thickness corresponds to \SI{4.6}{\micro \m}, which leads to LC response times in the range of milliseconds\footnote{https://www.merckgroup.com/en/news/liquid-crystal-smart-antennas-29-03-2021.html}.
The insertion loss corresponds to \SI{-6.6}{\dB}.
Furthermore, the prototype supports a bandwidth of \SI{6.8}{\GHz}~(10.9\%) under a -\SI{6}{\dB} criterion.
With each phase shifter featuring full 360° tunability, the RIS provides effective beamsteering within an angular range of $\pm 50$°.

A low static power consumption is verified using an E4980A Precision LCR Meter with a $1$ KHz sinusoidal signal.
We compute the average power consumption under the assumption of equal likelihood of all differential phases.
Consequently, power consumption was evaluated across varying bias voltages (i.e., phases).
The average power consumption of an individual delay line phase shifter in the proposed \gls{LC-RIS} is $\mathrm{\SI{21.5}{\nW}}$.
Hence, for a large RIS with $\mathrm{10^6}$ phase shifters (corresponding to a 5 x $\mathrm{\SI{5}{\m}}$ aperture for an element spacing of $\mathrm{\lambda_0/2}$), the static power consumption would add up to $\mathrm{\SI{21.5}{\mW}}$ for \sysname.
As an example for a PIN diode, the power consumption reported in~\cite{Tang2021} at $\mathrm{\SI{10.5}{\GHz}}$ equals $\mathrm{\SI{0.33}{\mW}}$ per diode in an ON-state.
Hence, a PIN diode-based RIS with $\mathrm{10^6}$ elements
where half of the diodes are on would therefore present a
$\mathrm{\SI{165}{\W}}$, $\mathrm{\SI{330}{\W}}$ or $\mathrm{\SI{495}{\W}}$ static power consumption to bias the PIN diodes in a 1-, 2-, and 3-bit configuration, respectively.

\bibliographystyle{IEEEtran}
\bibliography{IEEEabrv,ref}

\end{document}